\documentclass[sn-mathphys,Numbered]{sn-jnl}

\usepackage{graphicx}%
\usepackage{multirow}%
\usepackage{amsmath,amssymb,amsfonts}%
\usepackage{amsthm}%
\usepackage{mathrsfs}%
\usepackage[title]{appendix}%
\usepackage{xcolor}%
\usepackage{textcomp}%
\usepackage{manyfoot}%
\usepackage{booktabs}%
\usepackage{algorithm}%
\usepackage{algorithmicx}%
\usepackage{algpseudocode}%
\usepackage{listings}%

\RequirePackage{hyperref}

\usepackage{amsmath}
\usepackage{amsfonts}
\usepackage{amsthm}
\usepackage{eurosym}
\usepackage{graphicx}
\usepackage{amssymb}
\usepackage{epstopdf}
\usepackage{topcapt}
\usepackage{latexsym}
\usepackage{mathtools}
\usepackage{graphics}
\usepackage{subfigure}
\usepackage{dsfont}
\usepackage[english]{babel}
\usepackage{color}

\usepackage{float}

\newcommand{\ud}{\mathrm{d}}

\DeclareMathOperator*{\esssup}{ess\,sup}

\theoremstyle{thmstyleone}%
\newtheorem{theorem}{Theorem}
\newtheorem{lemma}[theorem]{Lemma}%

\theoremstyle{thmstyletwo}%
\newtheorem{remark}{Remark}%

\theoremstyle{thmstylethree}%

\begin{document}

\title[On hierarchical competition through reduction of individual growth]{On hierarchical competition through reduction of individual growth}


\author[1]{\fnm{Carles} \sur{Barril} }\email{carles.barril@uab.cat}

\author[1,3]{\fnm{\`{A}ngel}  \sur{Calsina}}\email{angel.calsina@uab.cat}

\author[2]{\fnm{Odo} \sur{Diekmann}}\email{O.Diekmann@uu.nl}

\author*[1]{\fnm{J\'{o}zsef Z. } \sur{Farkas}}\email{jozsefzoltan.farkas@uab.cat}

\affil*[1]{\orgdiv{Department de Matem\`{a}tiques}, \orgname{Universitat Aut\`{o}noma de Barcelona}, \orgaddress{\street{} \city{Bellaterra}, \postcode{08193},  \country{Spain}}}

\affil[2]{\orgdiv{Department of Mathematics}, \orgname{University of Utrecht}, \orgaddress{\street{Budapestlann 6, PO Box 80010}, \city{Utrecht}, \postcode{3508 TA}, \country{The Netherlands}}}

\affil[3]{\orgdiv{Centre de Recerca Matem\`{a}tica},  \orgaddress{ \city{Bellaterra}, \postcode{08193},  \country{Spain}}}


\abstract{We consider a population organised hierarchically with respect to size in such a way that the growth rate of each individual depends only on the presence of larger individuals. As a concrete example one might think of a forest, in which the incidence of light on a tree (and hence how fast it grows) is affected by shading by taller trees. The classic formulation of a model for such a size-structured population employs a first order quasi-linear partial differential equation equipped with a non-local boun\-dary condition.  However,  the model can also be formulated as a delay equation, more specifically a scalar renewal equation, for the population birth rate.  After discussing the well-posedness of the delay formulation,  we analyse how many stationary birth rates the equation can have in terms of the functional parameters of the model.  In par\-ti\-cular we show that,  under reasonable and rather general assumptions, only one stationary birth rate can exist besides the trivial one (associated to the state in which there are no individuals and the population birth rate is zero). We give conditions for this non-trivial stationary birth rate to exist and analyse its stability using the principle of linearised stability for delay equations.  Finally,  we relate the results to the alternative,  partial differential equation formulation of the model.}



\maketitle

\textit{We dedicate this paper to Professor Glenn F. Webb, a friend, mentor and distinguished scientist. Over the past 50 years Glenn has made tremendous contributions to a wide variety of research domains, ranging from semigroup theory to cancer modelling. Structured population dynamics has been a central fixture to Glenn's research interest for decades. Indeed, Glenn is recognised as one of the worldwide leading experts of age- and size-structured population dynamics, an exciting field of research, which has enjoyed tremendous growth in recent decades. We are happy to have the opportunity to make a small contribution to this field and to this special issue honouring Glenn and celebrating his achievements.}

\section{Introduction}\label{sec1}

In terms of numbers, the dynamics of a population is generated by mor\-ta\-lity and reproduction.   In structured population models  \cite{Webb1985, MD1986, MR2008, dull2022}, individuals are characterized by variables such as age, size or other (physiological) characteristics. In that case, development/maturation needs to be modelled too (a trivial task in the case of age, but certainly not in general!).

   As explained in detail in \cite{DGHKMT2001}, density dependence is most easily incorporated in a two step procedure: i) first introduce the environmental condition via the requirement that individuals are independent from one another when this condition is prescribed as a function of time; ii) next model feedback by specifying how, at any particular time, the environmental condition is influenced by the population size and composition. In the inspiring book \cite{RP2013} detailed ecological motivation is presented for including in this feedback loop the impact of density dependence on development and maturation. 

   Here our aim is to investigate in the context of a toy model the consequences of density dependence that only affects development directly (fertility is affected indirectly, since it depends on the developmental stage of the individual). We do so for a one-dimensional i-state (i.e., the variable capturing the relevant differences among individuals ‘lives’ on the real line), so for an i-state space that comes equipped with an order relation. In fact we shall assume that the presence of ‘larger’ individuals has a negative impact on the growth rate of ‘smaller’ individuals (as a motivating example one might think of trees and shading, with the i-state interpreted as ‘height’; but please note that we ignore spatial structure and that, consequently, the model is but a caricature).
For the incorporation of space into physiologically structured population models see \cite{webb2008}.  For an alternative approach to hierarchically structured models see \cite{thieme1986}.

The organisation of the paper is as follows. In Section 2 we first present the classic PDE formulation of the model. Then we present biological assumptions underlying the model and deduce a scalar nonlinear renewal equation for the population birth rate (the so called delay formulation). In Section 3 a dynamical systems framework for the renewal equation is outlined. In Section 4 we give conditions guaranteeing the existence of a non-zero stationary birth rate. In Section 5 we apply the principle of linearised stability for delay equations \cite{Diek1} to prove that, for a certain two-parameter family of fertility functions, such a stationary birth rate (whenever it exists) is locally asymptotically stable. We also show that, under natural hypotheses on the ingredients, the zero stationary birth rate is a global attractor when it is the only stationary birth rate.

In Appendix \ref{appendixB} a technical result needed in Section \ref{section-stability} is shown. 
In Appendix \ref{appendixPDE}  the more classical formulation of the model,  taking the form of a first order PDE involving non-local functionals, is studied. In particular we show that the conditions guaranteeing the existence of stationary population densities (with respect to height) coincide with the conditions guaranteeing non-trivial stationary birth rates in the delay formulation. This makes sense since both formulations model the same phenomena (although they are independently derived from biological assumptions).  Such a phenomenological relation between the two formulations suggests that the stability results for the delay formulation can be translated to the PDE formulation (as indeed is done in \cite{BCDF2022}). Although this issue is not addressed rigorously in the present paper, some comments are included in the concluding remarks section.

\section{Two different formulations of a structured population model}\label{sectionDelayForm}

The classical formulation of the model we study here is derived by imposing a conservation law that leads to the (non-local, quasi-linear and first-order) partial differential equation:
\begin{equation}\label{pde}
\begin{aligned}
\frac{\partial }{\partial t} u(x,t)  + \frac{\partial }{\partial x}\left(g(E(x,t))u(x,t) \right) + \mu u(x,t) & = 0, \\
 g(E(x_{m},t)) u(x_{m},t) & =  \displaystyle\int_{x_m}^{\infty} \beta(y) u(y,t)\,\ud y, \\
 E(x,t) & =\int_x^\infty u(y,t)\,\ud y.
   \end{aligned}
\end{equation}
Our motivation to study the specific model above is to understand whether the evolution (the interpretation of $x$ and $u$ is explained below) of a tree population can be explained by taking into account only competition for light through a hierarchical structure affecting individual growth,  assuming that resources (such as water,  space,  etc.) are readily available.  
Indeed, we assume that the growth rate $g$ of an individual of height $x$ does not depend on $x$ directly, but only indirectly, as it depends on the amount of light the individual receives per unit of time. We assume that the latter, in turn, is fully determined by the number $E(x,t)$ of individuals that are taller than $x$ (we call $E$ an interaction variable, since it mediates how the environmental condition, here light intensity, is influenced by the extant population). We assume that the per capita death rate $\mu$ and the per capita reproduction rate $\beta$ only depend on the height $x$. In fact we assume that $\mu$ is constant, i.e., independent of $x$, while $\beta$ is a non-decreasing function of $x$. We assume that all individuals are born with the minimal height $x_m$ and that $g$ is positive (we do not impose an upper bound on height). The assumptions that $\mu$ and $\beta$ do not depend on the environment $E$ allows us to derive fairly explicit stability criteria,  as we will show later on. 

In \eqref{pde} the second equation stands for the flux of newborns, offspring
of individuals of any size $y$ which have a size specific per capita
fertility (obviously nonnegative) denoted by $\beta$. Notice that the fertility is indeed indirectly affected by negative density dependence since a larger value of the environmental variable leads to a smaller size achieved by the individuals. From a dynamical point of view solutions of (\ref{pde}) can be seen as orbits $t\mapsto u(\cdot,t)$ in the space of integrable functions with respect to height, i.e.  in $L^1(x_m,\infty)$.

A more general model,  when both $\mu$ and $\beta$ are functions of size $x$, and  the (somewhat more general)   environmental interaction variable 
$$E(x,t)=\alpha\int_0^x u(y,t)\,\ud y+\int_x^M u(y,t)\,\ud y,\quad \alpha\in [0,1],$$ 
(but with finite maximal size $M$) was studied  in \cite{cushing1994, AD,AI,FH3},  and a very general model incorporating distributed recruitment in \cite{CS2}.  In \cite{kraev2001} the well-posedness of the above problem was proven by rewriting the system in terms of characteristic coordinates.  
 We note that the focus in \cite{AD,AI} was on numerical approximation of solutions of the hierarchical model.  On the other hand,  in \cite{FH3} the authors derived a formal linearisation of the model and studied regularity properties of the governing linear semigroup.  A characteristic equation was also deduced for the special case when neither the growth rate $g$ nor the mortality rate $\mu$  depend on the interaction variable $E$ ($\beta$ on the other hand does).  Note however that the linearisation and stability results established in \cite{FH3} were completely formal,  as the Principle of Linearised Stability has not been established for the PDE formulation \eqref{pde}.  This is the main reason why in the current work we employ a different formulation of the model.  
 
 Specifically,  to derive a delay formulation of the model,  we assume (as in the case of the PDE)  that individuals are fully characterized by a variable $x$,  taking values in $\mathbb{R}_+$. In general, $x$ is called i-state but here, for clarity, we call it ‘height’, the point being that we motivate our assumptions about interaction in terms of competition for light (this phenomenon is also addressed mathematically in \cite{kraev2001, zavala2007, magal2017}, among others).  We assume that a density function $u=u(x, t)$ exists such that the integral of $u$ with respect to the first variable over an interval gives the number of individuals with size within this interval at time $t$. This allows us to write
\begin{equation}\label{environmentalVariable}
E(x,t) = \int_{x}^\infty  u(s,t) ds,
\end{equation}
so that the height of an individual evolves according to
\begin{equation}\label{growthequation}
X'(t)=g(E(X(t),t)).
\end{equation}

Let $B(t)$ denote the population birth rate at time $t$. Then $B$ equals the influx at $x_m$, which originates from reproduction by the extant population:
\begin{equation}\label{birthrateFormula}
B(t)=\int_0^\infty \beta(y)u(y,t)dy.
\end{equation}
Let $n(t,\cdot)$ denote the age density.   We do not need to write a PDE and solve it in order to conclude that
\begin{equation}\label{agePopDensity}
n(t,a) = B(t-a) e^{- \mu a}.
\end{equation}
This allows us to rewrite (\ref{birthrateFormula}) as
\begin{equation}\label{birthRateFormula1}
B(t) =    \int_0^\infty \beta( S(a,t) ) B( t - a ) e^{- \mu a} da,
\end{equation}
with $S(a,t)$ specifying the height of an individual having age $a$ at time $t$ (and hence being born at time $t - a$).

We refer to section III.4 of \cite{MD1986}, entitled ``Integration along characteristics, transformation of variables, and the following of cohorts through time'',   for general considerations about switching between size- and age-densities. Here the situation is relatively simple,  since  individuals taller than you are exactly those that are older than you,  i.e., were born earlier than you.  Or,  in a formula
\begin{equation}\label{environmentFormula1}
E(x,t) = \int_\tau^\infty B(t - \alpha) e^{-\mu \alpha}  d\alpha,
\end{equation}
when $x=S(\tau,t)$.  (For intricacies arising when individuals of different age can have the same size see \cite{thieme1988}.)

Next note that an individual that was born at time $t - a$ has age $\tau$ at time $t - a + \tau$. The height $y=y(\tau) := S(\tau, t - a + \tau)$ of such an individual evolves according to
\begin{equation}\label{heightODELaw}
\begin{aligned}
\frac{dy}{d\tau}(\tau) & =  g( E( y(\tau), t - a + \tau))\\
& = g( E(S(\tau, t - a + \tau), t - a + \tau) ) \\
& = g\left( \int_\tau^\infty B(t - a + \tau - \alpha) e^{-\mu \alpha  }d\alpha\right).
\end{aligned}
\end{equation}
Noting that $y(0) = x_m$  we obtain by integration that 
\begin{equation}\label{heightFunction}
\begin{aligned}
S(a,t) & = y(a) \\
& = x_m + \int_0^a g\left(\int_\tau^\infty B(t - a + \tau - \alpha) e^{-\mu \alpha}  d\alpha\right) d\tau.
\end{aligned}
\end{equation}
Inserting (\ref{heightFunction}) into (\ref{birthRateFormula1}) we obtain
\begin{equation} \label{scalar}
B(t)= \int_0^{\infty} \beta \bigg( x_m + \int_0^a g \left(
\int_{\tau}^{\infty} e^{-\mu \alpha} B_t(\tau-a-\alpha) \mbox{d}
\alpha \right) \, \mbox{d}\tau \bigg)\,\, e^{-\mu a}B_t(-a) \,
\mbox{d}a,
\end{equation}
where
\begin{equation}
B_t(\theta):=B(t+\theta).
\end{equation}
Notice that (\ref{scalar}) can also be written as
\begin{equation}  \label{scalar2}
B(t)= \int_0^{\infty} \beta \left( x_m + \int_0^a g \left( e^{- \mu(\tau -a)}
\int_{a}^{\infty} e^{-\mu s} B_t(-s) \mbox{d}s
 \right) \, \mbox{d}\tau \right)\,\, e^{-\mu a}B_t(-a) \,
\mbox{d}a.
\end{equation}

\section{The dynamical systems framework}

Equation $(\ref{scalar2})$ provides the delay formulation of the model, which we are going to study here.   In the delay formulation the state variable is the population birth rate history $B_t:=B(t+\cdot)$, instead of the population density $u(\cdot,t)$ with respect to height. Specifically, one can consider the state space (of the weighted birth rate histories)
\begin{equation*}
\mathcal{X}= L_\rho^1(-\infty,0):= \left\{ \phi \in L_{loc}^1(- \infty, 0):\left\vert \left\vert \phi \right\vert  \right\vert _{\mathcal{X}} = \int_{-\infty}^0 e^{\rho s}\vert \phi(s)\vert  ds < \infty \right\},
 \end{equation*}
for some $\rho > 0$ (so $\mathcal{X}$ contains, in particular, constant functions, and therefore the possible steady states) and the delay equation $B(t) = \mathcal{F} (B_t)$ with $\mathcal{F}: \mathcal{X} \rightarrow \mathbb{R}$ defined by
 \begin{equation}\label{scalar3}
 \mathcal{F}(\phi) = \int_0^{\infty} \beta \left( x_m + \int_0^a g \left( e^{- \mu(\tau -a)}
 \int_{a}^{\infty} e^{-\mu s} \phi (-s) \,\mbox{d}s
 \right) \, \mbox{d}\tau \right)\,\, e^{-\mu a} \phi(-a) \,
 \mbox{d}a.
\end{equation}
We denote by $\mathcal{X}^+$ the standard positive cone of  $\mathcal{X}$. 

As discussed in \cite{Diek2, Diek1}, the delay equation $B(t) = \mathcal{F} (B_t)$, together with an initial history $B_0=\phi \in \mathcal{X}$, can be interpreted as an abstract Cauchy problem with a semilinear structure:
\begin{equation}\label{ACPdelay}
\left\{
\begin{aligned}
& \frac{d}{dt} \varphi(t) = A \varphi(t) + \mathcal{F}(\varphi(t))\delta_0\\
& \varphi(0)=\phi\in\mathcal{X} 
\end{aligned}
\right. ,
\end{equation}
where $A$ is the generator of the linear semigroup defined as $T_A(t)\phi := \phi(t+\cdot)\mathds{1}_-(t+\cdot)$. Notice that the mapping $t \mapsto T_A(t)\phi$ tells us how the population birth rate history would evolve without considering birth (and growth and mortality), since all these processes are summarised in the $\mathcal{F}$ function. This  setting makes it possible to analyse the well posedness of the problem and some dynamical properties by means of a generalised variation of constants equation. The standard variation of constants equation cannot be applied in a straightforward manner (as in \cite{pazy1983}) since the semilinear term of the problem (namely $\phi \mapsto\mathcal{F}(\phi)\delta_0$) does not take values in $\mathcal{X}$, but in the space of measures.

Here the theory included in the references mentioned above (\cite{Diek2} and \cite{Diek1}) applies provided that $\mathcal{F}$ is continuously differentiable, which is stated in the following theorem, and proved in Appendix \ref{appendixB}. We assume that $g$ is smoothly extended to the whole of $\mathbb{R}$, implying that the right hand side of \eqref{scalar3} is defined on the whole Banach space $\mathcal{X}$ (so even for non positive $\phi$). Of course negative birth rates do not have biological meaning, but they allow us to work on the whole space (recall that the positive cone of $L^1$ has empty interior).

\begin{theorem}\label{differentiability}
	 Assume that $g : \mathbb{R} \rightarrow \mathbb{R} $ and $\beta : \mathbb{R}^{+} \rightarrow \mathbb{R}$ have a bounded and globally Lipschitzian first derivative. Also assume that $g$ is bounded and positive  and that $\beta$ is non-negative.
	Then the map $\mathcal{F}:\mathcal{X}\rightarrow \mathbb{R}$ defined in (\ref{scalar3}) is continuously differentiable with bounded derivative provided that $ \rho < \mu/5.$
\end{theorem}

\begin{theorem}
\textbf{Existence and uniqueness} Under the hypotheses of the previous theorem, for any $\phi \in \mathcal{X},$ there exists a unique $B \in L_{loc}^1(\mathbb{R})$ such that $B(t) = \phi(t)$ for $t < 0$ and $B(t)$ satisfies \eqref{scalar2} for $t \geq 0$.  Moreover, $B$ belongs to the positive cone $\mathcal{X}^+$  of $L_{loc}^1(\mathbb{R})$  whenever $\varphi \in \mathcal{X}^{+}$. 
	\end{theorem}

\begin{proof}
	It is an immediate consequence of Theorem \ref{differentiability} (notice that a bounded derivative implies global Lipschitz continuity), Theorem 3.15 in \cite{Diek1} (which implies the equivalence between (\ref{scalar2}) and (\ref{ACPdelay})) and Theorem 2.2 in \cite{Diek1} (which implies the existence and uniqueness of mild solutions of (\ref{ACPdelay}) and the generation of a nonlinear semigroup $\Sigma(t;\phi)$ satisfying $\Sigma(t;\phi)=B_t$). The facts that the linear semigroups in Theorem 2.2 of \cite{Diek1} are positive and $\mathcal{F}$ maps the positive cone of  $\mathcal{X}$ to $\mathbb{R}^{+}$ imply that $B$ belongs to the positive cone whenever $\phi \in \mathcal{X}^{+}$.
\end{proof}

Let $B\in\mathbb{R}$ be a stationary population birth rate, i.e. $B$ satisfies $B=\mathcal{F}(\bar{B})$ where $\bar{B}\in \mathcal{X}$ is defined by $\bar{B}(\theta)=B$ for (almost) all $\theta\in(-\infty, 0)$. The following theorem determines the local stability of $\bar{B}$ in terms of properties of $D\mathcal{F}(\bar{B})$. Since $D\mathcal{F}(\bar{B})$ is a bounded linear operator from $\mathcal{X}$ to $\mathbb{R}$, the Riesz Representation Theorem implies that $D\mathcal{F}(\bar{B})$ can be written as
\[
D\mathcal{F}(\bar{B})\phi = \int_0^\infty k(s)\phi(-s)ds =: \left< k,\phi\right>
\]
with $k$ an element of the dual space of $\mathcal{X}$, represented by
\[
\mathcal{X}'= L_\rho^\infty(0,\infty):= \left\{ f \in L^\infty(0, \infty):\left\vert \left\vert  f \right\vert  \right\vert_{\mathcal{X'}} = \sup_{s\in(0,\infty)} e^{\rho s}\vert f(s)\vert < \infty \right\}.
\]

\begin{theorem}\label{StabilityTheorem} (Theorem 3.15 in \cite{Diek1}) Under the hypotheses of Theorem \ref{differentiability}, let $\bar{B}\in \mathcal{X}$ be a stationary state of (\ref{scalar2}) and let $k\in \mathcal{X}'$ represent $DF(\bar{B})$. Consider the characteristic equation
 \begin{equation}\label{charEqTheorem}
0=1-\hat{k}(\lambda),
\end{equation}
where $\hat{k}$ is the Laplace transform of $k$ (i.e. $\hat{k}(\lambda)=\int_0^\infty e^{-\lambda s}k(s)ds $). 
\begin{itemize}
\item[(a)] If all roots of the characteristic equation (\ref{charEqTheorem}) have negative real part, then the stationary state $\bar{B}$ is locally asymptotically stable. 
\item[(b)] If there exists at least one root with positive real part, then the steady state $\bar{B}$ is unstable.
\end{itemize}
\end{theorem}

\section{Existence and characterization of steady states}\label{section_EofSS}

A stationary solution of the problem can be found by simply assuming
that $B$ in (\ref{scalar2}) is independent of $t$. Of course there is a trivial
stationary solution $B=0$ that corresponds to the absence of
individuals. When dealing with non-trivial stationary solutions of \eqref{scalar2}, we make the following abuse of notation to ease readability: we use $\bar{B}$ to denote a constant function in $\mathcal{X}$ and $\bar{B}\in\mathbb{R}$ as the image it takes (so that we let the context tell whether $\bar{B}$ refers to the constant function or to the value it takes). With this in mind, and taking into account \eqref{scalar2}, it follows that a non-trivial equilibrium $\bar{B}\in\mathcal{X}$ is a constant function whose image is a non-zero solution of
\begin{equation}\label{def_R}
1= \int_0^{\infty} \beta \left( x_m + \int_0^a g \left( B
\frac{e^{-\mu \tau}}{\mu} \right) \, \mbox{d}\tau \right)\,\, e^{-\mu a}
\, \ud a =: R(B).
\end{equation}
Under natural hypotheses concerning $\beta$ and $g$, which
essentially amount to assuming that larger sizes correspond to larger
fertilities, that more competition (more individuals higher in the hierarchy than
the one we are observing) means slower growth, and that the first generation progeny of
an individual is finite (more precisely, that $R(0)<\infty$), we readily obtain the
following theorem.

\begin{theorem} \label{R0}
	Under the hypotheses of Theorem 1 and the assumptions that $\beta$ is a strictly increasing function on $[x_m,\infty)$, and that  $g$ is a strictly decreasing function on $[0, \infty)$, there exists a non-trivial equilibrium of \eqref{scalar2} if and only if
		\begin{equation*}
		\begin{aligned}
		R_0:= R(0) = &  \int_0^{\infty} \beta \big( x_m + g(0)a \big)\, e^{-\mu a}\,  \ud a >1\qquad\text{and}  \\ 
		 & \int_0^{\infty} \beta \big( x_m + g(\infty)a \big)\, e^{-\mu a}\,
		\ud a <1,
		\end{aligned}
		\end{equation*}
		and there is at most one such non-trivial equilibrium.
	\end{theorem}
\begin{proof} The hypotheses imply that $R$ is a well defined continuous and strictly decreasing function on $[0,\infty)$.  A double application of the Lebesgue dominated convergence theorem then gives
	\begin{equation*}
		\begin{aligned}
	 \lim_{R \rightarrow \infty} R(B) = & \int_0^{\infty} \beta \left(x_m + \int_0^a \displaystyle\lim_{B \rightarrow \infty} g\left(B \frac {e^{-\mu \tau}}{\mu}\right) \ud \tau  \right)e^{-\mu a} \ud a  \\ 
	 =&  \int_0^{\infty} \beta \left( x_m + g(\infty)a \right)\, e^{-\mu a}\,
	\ud a.	
	\end{aligned}
		\end{equation*}
\end{proof}

\begin{remark}
As usual,  $R_0$ can be interpreted as the so-called \textrm{basic
reproduction number}, i.e., the expected number of offspring
 individuals will produce during their lifetime in a  \textrm{virgin} environment, i.e.,
when there are no individuals older/larger than themselves in the population.
\end{remark}

\subsection{Age and size equilibrium profiles}\label{subsec2}

The age density of a steady state is given by $\bar{n}(a) = \bar{B}e^{-\mu a}$ (see (\ref{agePopDensity})). 
Let us now  set
\begin{equation} \label{size_equilibrium}
\bar{S}(a)= x_m +\int_0^{a} g \left(\int_{\tau}^{\infty} \bar{B}  e^{-\mu \alpha} \mbox{d} \alpha
 \right) \mbox{d} \tau = x_m +\int_0^{a} g \left( \bar{B} \frac {e^{-\mu \tau}}{\mu} \right) \mbox{d} \tau,
\end{equation}
 which is the size of an individual of age $a$ at the nontrivial equilibrium, see (\ref{heightFunction}).  
The density $\bar{u}(x)$ with respect to size of the same population distribution can then be computed by taking into account the equality
\begin{equation*}
\int_{\alpha_1}^{\alpha_2} \bar{n}(a)\, \ud a = \int_{\bar{S}(\alpha_1)}^{\bar{S}(\alpha_2)}  \frac{\bar{n}\big(\bar{S}^{-1}(x)\big)}{\bar{S}'\big(\bar{S}^{-1}(x) \big)} \, \ud x = \int_{\bar{S}(\alpha_1)}^{\bar{S}(\alpha_2)} \bar{u}(x)\, \ud x,
\end{equation*}
which follows from the change of variable $x=\bar{S}(a)$ and the interpretation of $\bar{n}$ and $\bar{u}$. Thus, we find
\begin{equation}\label{size_density}
\bar{u}(x) = \frac{\bar{n}\big(\bar{S}^{-1}(x)\big)}{\bar{S}'\big(\bar{S}^{-1}(x) \big)} = \frac{\bar{B}e^{-\mu \bar{S}^{-1}(x)}}{g\left(\bar{B}\frac{e^{-\mu \bar{S}^{-1}(x)}}{\mu}\right)},
\end{equation}
which is an alternative expression to (\ref{new_steady_state}).

\section{Stability of steady states}\label{section-stability}
The linearisation of \eqref{scalar} around the origin is simply (see \ref{differential} in Appendix \ref{appendixB}),
\begin{equation}\label{linSolAt0}
y(t) = D \mathcal{F}(0) y_t = \int_0^{\infty} \beta \big( x_m + g(0)a \big)\, e^{-\mu a}y(t-a) \, \mbox{d}a =: \int_0^{\infty} k(a) y_t(-a)\, \mbox{d}a
\end{equation}
(as indeed one can understand by using only the interpretation: it describes the linear population model corresponding to the virgin environment $E=0$).
\begin{theorem} \label{linear}
Under the hypotheses of Theorem \ref{differentiability}, the trivial
equilibrium is (locally) exponentially stable if $R_0<1$,  and it is unstable if $R_0>1.$
\end{theorem}
\begin{proof}
Clearly the kernel $k \in L_{\rho}^{\infty}(0, \infty)$ corresponds to the Riesz representation of $D\mathcal{F}(0)$. Then, according to Theorem \ref{StabilityTheorem}, the stability of the steady state is determined by the sign of the real part of the
zeroes of the characteristic equation $\hat{k}(\lambda)=1,$ where
$\hat{k}$ stands for the Laplace transform of $k.$ 
$\hat{k}(\lambda)$ is defined (at least) for $Re(\lambda) >-\rho.$ Moreover, since the kernel $k$ is positive, $\hat{k}$ is for real $\lambda$ a decreasing function with limit 0 at infinity. Hence there is at most
one real solution $\hat{\lambda}$ of the characteristic equation,
which indeed exists and is positive if $\hat{k}(0) = R_0 > 1.$ So then the trivial equilibrium is unstable.\\
When $\hat{k}(0) = R_0 < 1,$ if there is a real root, it is
negative.
Moreover, if a non-real $\lambda$ is a root of the characteristic equation, then  $1 = \hat{k}(\lambda) =Re (\hat{k}(\lambda)) < \hat{k} (Re \lambda)$, which implies, by the fact that $\hat{k}$ tends to $0$, that there is a real root $\hat{\lambda}$ larger than $Re(\lambda)$. As such a real root is necessarily negative, the trivial equilibrium is locally exponentially stable.
\end{proof}

\begin{theorem} \label{global}
If $R_0<1$ and the hypotheses of Theorem \ref{R0} hold,  then all solutions of (\ref{linSolAt0}) tend exponentially to $0$ as $t\rightarrow \infty$.
\end{theorem}

\begin{proof}
For a given solution let us write (cf. \eqref{heightFunction})
\begin{equation*}
S(a,t) = x_m + \int_0^a g \left( \int_{\tau}^{\infty} e^{-\mu \alpha}
B_t(\tau-a-\alpha) \,\mbox{d} \alpha \right) \, \mbox{d}\tau,
\end{equation*}
the size at time $t$ of an individual of age $a.$ From \eqref{scalar} we can write
\begin{equation}\label{gronwall}
\begin{aligned}
B(t) &= \int_0^{\infty} \beta(S(a,t)) e^{-\mu a} B(t-a) \,\mbox{d} a \\
& =
\int_{-\infty}^0 \beta(S(t-s,t)) e^{-\mu(t-s)} B(s) \,\mbox{d} s  +
\int_{0}^{t} \beta(S(t-s,t)) e^{-\mu (t-s)} B(s) \,\mbox{d} s \\ 
 &=: f(t) + \int_{0}^{t} \beta(S(t-s,t)) e^{-\mu (t-s)} B(s)\,\mbox{d} s \\
 & \leq
f(t) + \int_{0}^{t} \beta(x_m + g(0) (t-s)) e^{-\mu (t-s)} B(s)\,\mbox{d}
s.
\end{aligned}
\end{equation}
The kernel $k(a) = \beta(x_m + g(0) a) e^{-\mu a}$ of the linear Volterra integral equation
\begin{equation} \label{Volterra}
y(t) = f(t) + \int_0^t k(t-s) y(s) \,\ud s
\end{equation}
 has a nonnegative resolvent $r$ (meaning that $r(t) = k(t) + \displaystyle\int_0^t k(t-s)r(s) \,\ud s$ and $y(t) = f(t) + \displaystyle\int_0^t r(t-s)f(s) \,\ud s$) (see Theorem 2.3.4 in \cite{Gripen}). Then by a generalized Gr\"{o}nwall lemma, one obtains $B(t)\leq
y(t)$ where $y(t)$ is the solution of \eqref{Volterra}.

Indeed, using the usual notation for convolution, \eqref{gronwall} can be written as $B \leq f + k * B$ and so $B = f-g + k*B$ for $g = f + k*B -B  \geq 0.$ Then we have
\begin{equation*}
B = f - g + r*(f - g) = f + r*f - (g + r*g ) =y-(g+r*g)\;\Rightarrow\; B\leq y,
\end{equation*}
since $r$ and $g$ are non-negative (cf. Theorem 9.8.2 in \cite{Gripen}).
The claim follows since $y(t)$ tends exponentially to $0$ when $R_0 < 1$ by Theorem 3.12 in \cite{Diek1} and the final part of the proof of Theorem \ref{linear}.
\end{proof}

Let us recall the notation
\begin{equation} \label{size_equilibrium}
\bar{S}(a)= x_m +\int_0^{a} g \left( \bar{B} \frac {e^{-\mu \tau}}{\mu} \right) \,\mbox{d} \tau,
\end{equation}
for the size of an individual of age $a$ at the non-trivial equilibrium.

Let us now compute the linearisation of \eqref{scalar} around the
nontrivial equilibrium $\bar{B}$ using \eqref{size_equilibrium}. For this we set $B(t) = \bar{B} + y(t)$ and write (formally)
\begin{equation}
\begin{array}{ll}
\bar{B} + y(t)\\
= \displaystyle\int_0^{\infty} \left( \beta \big( \bar{S}(a)\big) + \beta' \big(
\bar{S}(a)\big)\int_0^a  g'\left(\bar{B} \frac{e^{-\mu\tau}}{\mu}\right) \right. \\
\left.  \quad\quad\quad \quad  \times\displaystyle\int_{\tau}^{\infty} e^{- \mu \alpha} y_t(-a+\tau-\alpha) \mbox{d} \alpha \mbox{d} \tau + o(y_t) \right) e^{-\mu a} \big(\bar{B} + y_t(-a)\big) \,\mbox{d} a,
\end{array}
\end{equation}
which, using the steady state condition \eqref{def_R} and neglecting higher order terms, leads to
\begin{equation} \label{nontrlinear}
\begin{aligned}
y(t)  & =  D \mathcal{F}(\bar{B}) y_t =  \int_0^{\infty} \beta \big( \bar{S}(a) \big)\, e^{-\mu a}y(t-a)\,
\mbox{d}a \\
  & +  \displaystyle\int_0^{\infty} \beta'\big( \bar{S}(a) \big)\, e^{-\mu a}
\bigg( \, \int_0^a \bar{B} g'\left(\bar{B} \frac{e^{-\mu
\tau}}{\mu}\right) \int_{\tau}^{\infty} e^{- \mu \alpha}
y(t-a+\tau-\alpha)\, \mbox{d} \alpha \mbox{d} \tau \bigg)\, \mbox{d} a\\
&  = \int_0^{\infty} \beta \big( \bar{S}(a) \big) e^{-\mu a}y(t-a)\,  \mbox{d}a \\
 & + \int_0^{\infty} \beta'\big(\bar{S}(a) \big) e^{-\mu a}  \left(\int_0^a \bar{B} g'\left(\bar{B} \frac{e^{-\mu \tau}}{\mu}\right) e^{- \mu (\tau-a)}
\int_{a}^{\infty} e^{- \mu \sigma} y(t-\sigma) \mbox{d}
\sigma \, \mbox{d} \tau \right) \mbox{d} a.
\end{aligned}
\end{equation}

\begin{remark}
See Appendix \ref{appendixB} for a rigorous derivation of (\ref{nontrlinear}). There, $\mathcal{F}$ is written essentially as the composition of simpler functions and then the chain rule is applied.  
\end{remark}

Changing the order of integration, the expression within parentheses inside the last integral in (\ref{nontrlinear})
can be rewritten as:
\begin{equation}
\begin{aligned}
& \int_0^a \bar{B} g'\left(\bar{B} \frac{e^{-\mu \tau}}{\mu}\right)
\int_{a}^{\infty} e^{- \mu (-a+\tau + \sigma)} y(t-\sigma) \,\mbox{d}
\sigma\,  \mbox{d} \tau \\
=  & \int_a^{\infty} \int_{0}^{a} \bar{B} g'\left(\bar{B} \frac{e^{-\mu
\tau}}{\mu}\right)  e^{- \mu \tau} \mbox{d} \tau e^{- \mu (\sigma -
a)} y(t-\sigma)\, \mbox{d}\sigma \\ 
=  & \int_a^{\infty} \left(
g\left(\frac{\bar{B}}{\mu}\right)- g\left(\bar{B} \frac{e^{-\mu
a}}{\mu}\right) \right) e^{- \mu (\sigma - a)} y(t-\sigma)\,
\mbox{d}\sigma.
\end{aligned}
\end{equation}
Thus, changing the integration order again, the second term on the
right hand side of \eqref{nontrlinear} reads
$$
\int_0^{\infty} \int_0^{\sigma} \beta'(\bar{S}(a)) \left[
g\left(\frac{\bar{B}}{\mu}\right)- g\left(\bar{B} \frac{e^{-\mu a}}{\mu}\right) \right]
 \mbox{d}a \, e^{-\mu \sigma} \, y(t-\sigma)\, \mbox{d} \sigma.
$$
Hence, \eqref{nontrlinear} is of the form $y(t)=\displaystyle\int_0^{\infty} k(a)
y(t-a) \,\mbox{d} a$ with the kernel
$$
k(a) = \beta \big( \bar{S}(a) \big)\, e^{-\mu a} +  e^{-\mu a}
\int_0^{a} \beta'(\bar{S}(\alpha)) \left[g\left(\frac{\bar{B}}{\mu}\right) -
g\left(\frac{\bar{B} e^{-\mu \alpha}}{\mu}\right) \right]\, \mbox{d} \alpha.
$$
Since
\begin{equation*}
\begin{aligned}
& \int_0^a \beta'(\bar{S}(\alpha)) g\left(\frac{\bar{B} e^{-\mu \alpha}}{\mu}\right)\mbox{d} \alpha = \int_0^a \beta'(\bar{S}(\alpha)) \bar{S}'(\alpha)\, \mbox{d} \alpha \\
 =  &\beta( \bar{S}(a))-\beta( \bar{S}(0))=\beta( \bar{S}(a))-\beta( x_m),
 \end{aligned},
\end{equation*}
the kernel $k$ simplifies to
\begin{equation*}
k(a)=\beta(x_m)e^{-\mu a}+g\left(\frac{\bar{B}}{\mu}\right)e^{-\mu a}\int_0^a \beta'(\bar{S}(\alpha))\, \mbox{d} \alpha,
\end{equation*}
which leads to the characteristic equation
\begin{equation}\label{DE-char-new_1}
1=\hat{k}(\lambda)=\frac{\beta(x_m)}{\lambda+\mu}+\frac{1}{\lambda+\mu} g\left(\frac{\bar{B}}{\mu}\right) \int_0^\infty \beta'(\bar{S}(a)) e^{-( \lambda+\mu)a}\,\ud a.
\end{equation}

Without the loss of generality, we will assume in the rest of this section that the minimum size is $x_m=0$.  This (technical) assumption makes the notation simpler.

\subsection{Sufficient conditions for stability of the non-trivial steady state}

 In this section we focus on the situation described in Theorem \ref{R0},   with both inequalities satisfied, so the situation in which a unique non-trivial steady state exists. Our aim is to formulate additional conditions on $g$ and $\beta$ that guarantee that all roots of the characteristic equation have negative real part, allowing us to conclude from Theorem \ref{StabilityTheorem}(a) that the non-trivial steady state is asymptotically stable. The strategy is to first consider a piecewise linear $\beta$ for which we can verify explicitly that all roots are in the (open) left half plane. Next we use Rouch\'{e}'s Theorem to show that for a large class of smooth perturbations of the piecewise linear $\beta$ the property 'no roots in the (closed) right half plane' is preserved.
   As already indicated, we fix $g$ but consider $\beta$ as a variable. Therefore the steady state value $\bar{B}$ (the solution of $R(B) = 1$) depends on $\beta$ and we shall use the notation $\bar{B}(\beta)$.  
To simplify the exposition we assume that 
\begin{equation}\label{g-assumption}
\lim_{z\to\infty}g(z)=0,
\end{equation}
which implies that the second inequality in Theorem \ref{R0} holds if $\beta(0)=0$ (recall that $x_m=0$).

Let us first assume  that the per capita fertility is
given by 
\begin{equation}
\beta_0(s):= \beta_{00} \max\{0,s-x_A\},\label{special-beta}
\end{equation}
where $x_A \geq 0$ is the adult size at which individuals start to reproduce.  Note that $\beta_0$ is not $C^1$ for $x_A>0$.  We believe that $\mathcal{F}$ defined by \eqref{scalar3}
 is nevertheless $C^1$, since $\beta_0$ is continuous.  We refrain from an attempt to prove this, as such involves, no doubt, many technicalities.   
 
 First we compute  (via integration by parts)
\begin{equation*}
\begin{aligned} 
R_0 = R(0) = & \int_0^\infty \beta_0\left( \int_0^a g(0) \mbox{d} \tau\right) e^{-\mu a}\mbox{d} a =  \int_0^\infty \beta_{00} \max\{ 0, g(0)a-x_A\}   e^{-\mu a}\mbox{d} a \\ 
= &\beta_{00} \int_{\frac{x_A}{g(0)}}^\infty  (g(0)a-x_A)   e^{-\mu a}\mbox{d} a  =  \frac{\beta_{00} g(0)}{\mu^2} e^{- \frac{\mu x_A}{g(0)} }.
\end{aligned}
\end{equation*}
Recall that a non-trivial steady state exists if and only if $R_0>1$.  So we assume in the rest of the subsection that $\beta_{00}$ and $x_A$ are such that this assumption holds.

Assume that the function $g$ is strictly decreasing.  Define $\bar{a}$ by
\begin{equation}\label{abar}
\int_0^{\bar{a}} g\left(\frac{\bar{B}(\beta_0) e^{- \mu \tau}}{\mu}\right) \mbox{d} \tau =x_A,
\end{equation}
i.e., $\bar{a}$ is the age at which individuals begin to reproduce given the environmental condition associated to the equilibrium. 

We have from \eqref{def_R} and with $x_m=0$ that
\begin{equation}\label{R_function}
\begin{aligned}
1=R(\bar{B}(\beta_0)) = & \int_0^{\infty} \beta_0 \left(\int_0^a g \left(\frac{\bar{B}(\beta_0)
	e^{-\mu \tau}}{\mu} \right) \, \mbox{d}\tau \right)\,\, e^{-\mu a} \,
\mbox{d}a \\
 =  & \beta_{00} \int_0^{\infty} \max\left\{0,\int_0^{a} g \left(\frac{\bar{B}(\beta_0) e^{-\mu
			\tau}}{\mu} \right) \mbox{d} \tau - x_A \right\}  e^{-\mu a} \mbox{d} a  \\
= & \beta_{00} \int_{\bar{a}}^{\infty} \left( \int_{0}^{a} g \left(\frac{\bar{B}(\beta_0) e^{-\mu
		\tau}}{\mu} \right) \mbox{d} \tau - \int_{0}^{\bar{a}} g \left(\frac{\bar{B}(\beta_0) e^{-\mu
		\tau}}{\mu} \right) \mbox{d} \tau \right) e^{-\mu a} \mbox{d} a
  \\
 = & \beta_{00} \int_{\bar{a}}^{\infty}\left( \int_{\bar{a}}^{a} g \left(\frac{\bar{B}(\beta_0) e^{-\mu
		\tau}}{\mu} \right) \mbox{d} \tau\right)  e^{-\mu a} \mbox{d} a \\
	 = &\beta_{00} \int_{\bar{a}}^{\infty} \left(\int_{\tau}^{\infty} e^{-\mu a} \mbox{d} a\right)
g \left(\frac{\bar{B}(\beta_0) e^{-\mu \tau}}{\mu} \right) \mbox{d} \tau \\
=& \beta_{00}
\int_{\bar{a}}^{\infty} \frac{e^{-\mu \tau}}{\mu}g \left(\frac{\bar{B}(\beta_0) e^{-\mu
		\tau}}{\mu} \right) \mbox{d} \tau = \frac{\beta_{00}}{\mu^2}\int_0^{e^{-\mu \bar{a}}} g\left(\frac{\bar{B}(\beta_0)}{\mu} \zeta\right) \ud\zeta.
\end{aligned}
\end{equation}

Next we note that  the characteristic equation \eqref{DE-char-new_1}
 reduces to
\begin{equation} \label{part-char}
	1 = \beta_{00} \frac{g\left(\frac{\bar{B}(\beta_0)}{\mu}\right)}{\lambda + \mu} \int_{\bar{a}}^{\infty} e^{-(\lambda + \mu) a} \mbox{d}a = \beta_{00} \frac{g\left(\frac{\bar{B}(\beta_0)}{\mu}\right)}{(\lambda+\mu)^2} e^{-(\lambda + \mu) \bar{a}}, 
\end{equation}
which, in the special case $x_A = 0$ (or, equivalently, $\bar{a}=0$) allows to identify the (two) roots as 
\begin{equation}\label{eigen}
	\lambda = - \mu \pm \sqrt{\beta_{00} g\left(\frac{\bar{B}(\beta_0)}{\mu}\right)} = \mu \left(-1 \pm \sqrt{R_0 \frac{g\left(\frac{\bar{B}(\beta_0)}{\mu}\right)}{g(0)}}\, \right).   
\end{equation}
So, under this assumption, we are able to explicitly formulate the
characteristic equation and even to explicitly compute its
roots.  From the condition for the existence of a nontrivial equilibrium \eqref{def_R} 
and \eqref{R_function} with $\bar{a} = 0,$ we have
\begin{equation}\label{g-estimate}
1 = R(\bar{B}(\beta_0)) = \frac{\beta_{00}}{\mu^2}\int_0^1g\left(\frac{\bar{B}(\beta_0)}{\mu} \zeta\right) \ud\zeta> \frac{\beta_{00}}{\mu^2} \min_{\zeta \in[0,1]} g\left(\frac{\bar{B}(\beta_0)}{\mu} \zeta\right) = \frac{\beta_{00}}{\mu^2} g\left(\frac{\bar{B}(\beta_0)}{\mu}\right),
\end{equation}
which implies that both eigenvalues in \eqref{eigen} are negative.  Hence,  Theorem \ref{StabilityTheorem}
 ensures that under these hypotheses (in particular that $x_A = 0$),  the nontrivial steady state is  locally asymptotically stable.  

Next we consider the case when $\bar{a}>0$,  i.e.  $x_A>0$.  Note that in this more general case,  similarly to \eqref{g-estimate},  we obtain the following estimate.
\begin{equation}\label{g-estimate2}
1 = R(\bar{B}(\beta_0)) = \frac{\beta_{00}}{\mu^2}\int_0^{e^{-\mu\bar{a}}} g\left(\frac{\bar{B}(\beta_0)}{\mu} \zeta\right) \ud\zeta> \frac{\beta_{00}}{\mu^2} e^{-\mu\bar{a}} g\left(\frac{\bar{B}(\beta_0)}{\mu}\right).
\end{equation}

Recall that in this case $\beta_0$ is not continuously differentiable.   However,  for any fixed value of $\beta_{00}$ and $x_A$   the roots of the characteristic equation have negative real parts,   see the proof of Theorem \ref{estimate-theorem}.  

 Hence the idea is to consider smooth fertility functions close enough (see Theorem \ref{estimate-theorem} 
below) to $\beta_0$,  and apply Theorem \ref{StabilityTheorem} to establish local asymptotic stability of the non-trivial steady state. 

Let us define
\begin{equation}
R(\beta,B)=\int_0^{\infty} \beta \left(\int_0^a g \left(\frac{B e^{- \mu \tau}}{\mu}\right)d\tau \right) e^{- \mu a} \ud a,
\end{equation}
and let us restrict to fertility functions such that $R(\beta,0) > 1.$  In this case there exists (a stationary birth rate)  $\bar{B}(\beta)>0$,   such that  $R(\beta,\bar{B}(\beta)) = 1$. 

\begin{lemma}\label{lemma} Let $b_1\geq \beta_{00}$. Let $g$ and $\beta$ satisfy the hypotheses of Theorems \ref{differentiability} and  \ref{R0},   let $g'(s)$ be strictly negative,  assume that $\displaystyle\lim_{z\to\infty}g(z)=0$ and $\beta(0)=0$, $\|\beta'\|_\infty \leq b_1$.  Then,  there exists a positive constant $C$,  only depending on $g$, $\mu$, $\beta_{00}$, $x_A$ and $b_1$, such that
  \begin{equation}
  \vert \bar{B}(\beta)-\bar{B}(\beta_0) \vert \leq C  \vert \vert \beta - \beta_0 \vert \vert_{\infty}.
  \end{equation}
\end{lemma}

\begin{proof}
We first obtain an upper bound on $\bar{B}(\beta).$ Let us define $\beta_1(s) = b_1 s$ and let $\tilde{B}$ be such that $R(\beta_1,\tilde{B}) = 1.$ Notice that the existence of $\tilde{B}$ is guaranteed because $R(\beta_1,0)>R(\beta_0,0) >1.$ We then have 
\begin{equation*}
R(\beta_1 ,\tilde{B}) = 1 = R(\beta,\bar{B}(\beta))  \leq R(\beta_1,\bar{B}(\beta)),
\end{equation*}
since $R$ is increasing with respect to its first argument (understanding that $\beta \leq \beta_1$ if $\beta(s) \leq \beta_1(s)$ for all $s\geq 0$).   Also note that $\bar{B}(\beta)) \leq \tilde{B}$
 because $R$ is decreasing with respect to its second argument.
 
Let us define $\tilde{a}$ such that $\displaystyle\int_0^{\tilde{a}} g\left(\frac{\tilde{B} e^{- \mu \tau}}{\mu}\right) \ud \tau = x_A .$ Since $g$ is decreasing, we have
\begin{equation}\label{atilde}
\int_0^{\tilde{a}} g\left(\frac{\bar{B}(\beta) e^{- \mu \tau}}{\mu}\right) \ud \tau \geq x_A. 
\end{equation}
Now  we can write
\begin{equation*}
\begin{aligned}
0&=R(\beta,\bar{B}(\beta)) -R(\beta_0,\bar{B}(\beta)) + R(\beta_0,\bar{B}(\beta)) -R(\beta_0,\bar{B}(\beta_0)) \\
& = R(\beta-\beta_0,\bar{B}(\beta)) +  R(\beta_0,\bar{B}(\beta)) -R(\beta_0,\bar{B}(\beta_0)).
\end{aligned}
\end{equation*}
Hence we have 
\begin{equation*}
\begin{aligned}
\frac{1}{\mu}  \vert \vert \beta - \beta_0 \vert \vert_{\infty} & \geq \vert R(\beta-\beta_0,\bar{B}(\beta)) \vert =  \vert R(\beta_0,\bar{B}(\beta)) - R(\beta_0,\bar{B}(\beta_0)) \vert \\
& = \beta_{00} \left\vert  \int_0^{\infty} \left( \max \left\{ 0, \int_0^a g\left(\frac{\bar{B}(\beta) e^{-\mu \tau}}{\mu}\right) \ud \tau - x_A  \right\} \right. \right. \\
& \qquad\qquad\qquad\quad \left. \left. -\max \left\{0 , \int_0^a g\left(\frac{\bar{B}(\beta_0) e^{-\mu \tau}}{\mu}\right) \ud \tau - x_A    \right\}\right) e^{- \mu a} \ud a   \right\vert  \\
& \geq \beta_{00} \int_{\tilde{a}}^{\infty} \int_0^a \left\vert  g\left(\frac{\bar{B}(\beta)e^{-\mu \tau}}{\mu}\right) - g\left(\frac{\bar{B}(\beta_0) e^{-\mu \tau}}{\mu}\right) \right\vert d\tau e^{-\mu a} \ud a,
\end{aligned}
\end{equation*}
where we used \eqref{atilde} and that the sign of $g\left(\frac{\bar{B}(\beta)e^{-\mu \tau}}{\mu}\right) - g\left(\frac{\bar{B}(\beta_0) e^{-\mu \tau}}{\mu}\right)$ only depends on the sign of $\bar{B}(\beta) - \bar{B}(\beta_0)$ due to the monotonicity of $g$.

So  we have (using $\frac{1}{\mu^3}e^{-\mu\tilde{a}}-\frac{1}{3\mu^3}e^{-2\mu\tilde{a}}\ge \frac{1}{2\mu^3}e^{-\mu\tilde{a}}$)
\begin{equation*}
\begin{aligned}
\frac{1}{\mu}  \vert \vert \beta - \beta_0 \vert \vert_{\infty} & \geq \beta_{00} \min_{s \in [0, \tilde{B}/\mu]} \vert g'(s) \vert  \int_{\tilde{a}}^{\infty} \left( \int_0^a \frac{e^{-\mu \tau}}{\mu} \ud\tau\right) e^{-\mu a} \ud a \, \vert \bar{B}(\beta) - \bar{B}(\beta_0) \vert \\
& \geq \beta_{00} \frac{e^{-\mu \tilde{a}}}{2 \mu^3} \min_{s \in [0, \tilde{B}/\mu]} \vert g'(s) \vert  \,  \, \vert \bar{B}(\beta) - \bar{B}(\beta_0) \vert,
\end{aligned}
\end{equation*}
and the claim follows immediately.
\end{proof}

We can now state
\begin{theorem}\label{estimate-theorem}
Let $b_1\geq \beta_{00}$. Let $g$ and $\beta$ satisfy the hypotheses of Theorems \ref{differentiability} 
and \ref{R0},   let $g'(s)$ be strictly negative,   assume that $\displaystyle\lim_{z\to\infty} g(z)=0$ and $\beta(0)=0$, $\|\beta'\|_\infty \leq b_1$.   Let $\bar{a}$ be defined implicitly through  \eqref{abar}. 
  Define
\begin{equation}\label{rMdefs}
r:=\frac{2 g(0) b_1}{\mu}\qquad\text{and}\qquad M:=\frac{\mu^2}{r^2}\left(1-\frac{\beta_{00}}{\mu^2} g\left(\frac{\bar{B}(\beta_0)}{\mu}\right)e^{-\bar{a}\mu}\right)>0,
\end{equation}
where positivity of $M$ follows from \eqref{g-estimate2}.   Moreover,   assume that
\begin{equation}\label{tildebetaCond}
\frac{g(0)}{\mu^2} \vert \vert \beta' - \beta_0' \vert \vert_{\infty} + \frac{\beta_{00}}{\mu^3}e^{-\mu \bar{a}} \vert \vert g' \vert \vert_{\infty} C \vert \vert \beta - \beta_0 \vert \vert_{\infty} < M,
\end{equation}
with $C$ given by Lemma \ref{lemma}.
Then, the unique positive stationary birth rate of the renewal equation \eqref{scalar2},   with fertility function $\beta$,   is locally asymptotically stable.
\end{theorem}

\begin{proof}
We denote the right hand side of \eqref{DE-char-new_1}  minus $1$ by $f(\lambda,\beta)$.   Note that $f(\cdot,\beta)$ is a holomorphic function of $\lambda$ for $Re(\lambda)>-\mu$,  which depends on the functional parameter $\beta$.  Note that $f(\lambda,\beta_0)$ equals the right hand side of  \eqref{part-char} minus $1$.

From \eqref{DE-char-new_1}  we have,   for any zero $\lambda$ of $f(\lambda,\beta)$ and of $f(\lambda,\beta_0)$,  such that $Re(\lambda)>-\mu/2$ the following uniform bound
\begin{equation*}
\vert \lambda+\mu\vert \leq g(0)b_1\int_0^\infty e^{-(Re(\lambda)+\mu)a} \ud a=\frac{g(0)b_1}{Re(\lambda)+\mu}<\frac{2 g(0) b_1}{\mu}=r.
\end{equation*}

Therefore,   all possible zeroes of $f(\lambda,\beta)$ and of $f(\lambda,\beta_0)$,  with non-negative real part,   belong to the (compact) disk segment 
\begin{equation*}
U:=\{\lambda\in\mathbb{C} \;:\; \vert \lambda+\mu\vert \leq r,\; Re(\lambda)\geq 0\}.
\end{equation*}
On the other hand,  $\vert f(\lambda,\beta_0)\vert >M$ on $U$ where $M$ is defined in \eqref{rMdefs} (so in particular  $f(\lambda,\beta_0)$ has no zero with non-negative real part).  Indeed,  using the triangle inequality,   we have
\begin{equation*}
\begin{aligned}
\vert f(\lambda,\beta_0)\vert &=\left\vert 1-\beta_{00} \frac{g\left(\frac{\bar{B}(\beta_0)}{\mu}\right)}{(\lambda+\mu)^2} e^{-(\lambda + \mu) \bar{a}}\right\vert =\frac{1}{\vert \lambda+\mu\vert^2}\left\vert (\lambda+\mu)^2-\beta_{00} g\left(\frac{\bar{B}(\beta_0)}{\mu}\right) e^{-\bar{a}(\lambda+\mu)}\right\vert  \\
& \ge  \frac{1}{\vert \lambda+\mu\vert^2} \left\vert Im(\lambda)^2+ \left(Re(\lambda)+\mu \right)^2-\beta_{00} g\left(\frac{\bar{B}(\beta_0)}{\mu}\right) e^{-\bar{a}(Re(\lambda)+\mu)}\right\vert. 
\end{aligned}
\end{equation*}
Noting that for $Re(\lambda)\ge 0$,  from \eqref{g-estimate2}  we have 
\begin{equation*}
 \left(Re(\lambda)+\mu\right)^2-\beta_{00} g\left(\frac{\bar{B}(\beta_0)}{\mu}\right) e^{-\bar{a}(Re(\lambda)+\mu)}\ge 0,
\end{equation*}
we conclude that there is no need for the absolute value bars around the second factor at the right hand side, and we obtain that
\begin{equation*}
\begin{aligned}
\vert f(\lambda,\beta_0)\vert   \geq \frac{\mu^2}{r^2}\left(1-\frac{\beta_{00}}{\mu^2} g\left(\frac{\bar{B}(\beta_0)}{\mu}\right) e^{-\bar{a}\mu}\right)=M.
\end{aligned}
\end{equation*}
Then,  Rouch\'{e}'s theorem ensures that any holomorphic function $h$ such that $\vert h(\lambda)-f(\lambda,\beta_0)\vert <M$ on the boundary of $U$,   does not vanish on $U$,  because $f(\lambda,\beta_0)$ does not vanish on $U$.   As a consequence, all zeros of $f(\lambda,\beta)$ will have negative real part whenever $ \vert f(\lambda,\beta)-f(\lambda,\beta_0)\vert <M$ on the boundary of $U$, which implies the stability of the non-trivial steady state corresponding to $\beta$ (by Theorem  \ref{StabilityTheorem}). 
  To this end note that one can show that $\vert f(\lambda,\beta)-f(\lambda,\beta_0)\vert $ is bounded above by the left hand side of \eqref{tildebetaCond}. Indeed,   using Lemma \ref{lemma},  and denoting by $\mathcal{H}$ the Heaviside step function,   we have for $Re(\lambda)\ge 0$ 
\begin{equation*}
\begin{aligned}
\vert f(&\lambda,\beta) - f(\lambda,\beta_0) \vert \\
= & \bigg\vert \frac{1}{\lambda + \mu} \int_0^{\infty} g\left(\frac{\bar{B}(\beta)}{\mu}\right) \left( \beta'(\bar{S}(a;\beta)) - \beta_{00} \mathcal{H}(a-\bar{a})  \right)e^{-(\lambda + \mu)a} \ud a \\
& + \frac{1}{\lambda + \mu} \int_0^{\infty} \left( g\left(\frac{\bar{B}(\beta)}{\mu}\right) -g\left(\frac{\bar{B}(\beta_0)}{\mu}\right)  \right) \beta_{00} \mathcal{H}(a-\bar{a}) e^{-(\lambda + \mu)a} \ud a \bigg\vert \\
\leq & \frac{g(0)}{\mu}\int_0^\infty  \vert \beta'(\bar{S}(a;\beta)) -\beta_{00} \mathcal{H}(a-\bar{a}) \vert e^{-\mu a} \ud a + \frac{\beta_{00}}{\mu^2} e^{-\mu\bar{a}} \left\vert g\left(\frac{\bar{B}(\beta)}{\mu}\right)-g\left(\frac{\bar{B}(\beta_0)}{\mu}\right)\right\vert& \\
\leq & \frac{g(0)}{\mu^2} \vert \vert \beta' - \beta_0' \vert \vert + \frac{\beta_{00}}{\mu^3} e^{- \mu \bar{a}} \vert \vert g' \vert \vert_{\infty}  \vert \bar{B}(\beta) - \bar{B}(\beta_0)\vert& \\
 \leq & \frac{g(0)}{\mu^2} \vert \vert \beta' - \beta_0' \vert \vert + \frac{\beta_{00}}{\mu^3} e^{- \mu \bar{a}} \vert \vert g' \vert \vert_{\infty} C  \vert \vert \beta - \beta_0 \vert \vert_{\infty}.
\end{aligned}
\end{equation*}
Therefore,  if $\beta$ satisfies \eqref{tildebetaCond}, the positive stationary birth rate is locally asymptotically stable.
\end{proof}

\begin{remark} Assumption  \eqref{tildebetaCond}
 can be replaced by a more explicit one using the proof of Lemma \ref{lemma} (i.e.  using the definitions of $\tilde{B}$ and $\tilde{a}$):
\begin{equation*}
\frac{g(0)}{\mu^2} \vert \vert \beta' - \beta_0' \vert \vert_{\infty} + \frac{2}{\mu}e^{\mu(\tilde{a}- \bar{a})} \frac{\vert \vert g' \vert \vert_{\infty}}{\displaystyle\min_{s\in [0, \tilde{B}/\mu]}{\vert g'(s) \vert}} \vert \vert \beta - \beta_0 \vert \vert_{\infty} < M.
\end{equation*}
\end{remark}

\subsection{Semi-explicit expression for a particular case}\label{part-cases}

In this section we assume that the per capita fertility is proportional to the size,   that is  $\beta(s) = \beta_0 s$ (i.e., $x_A =0$). In addition we consider that the individual growth rate is of the form $g(z) = \frac{g_0}{1+z/z_0}$ where $g_0 > 0$ and $z_0>0$ (recall that $z$ represents the environment that an individual experiences, which is given by the number of individuals that are larger than it). 

In this situation, \eqref{R_function} gives
$$
R(B) = \frac{\beta_0}{\mu^2} \int_0^{1} \frac{g_0}{1+\frac{B}{\mu} \frac{\zeta}{z_0}} \mbox{d} \zeta =  \frac{\beta_0 g_0}{\mu^2} \frac{\ln\left(1+B/(\mu z_0)\right)}{B/(\mu z_0)}= R_0 \frac{\ln\left(1+B/(\mu z_0) \right)}{B/(\mu z_0)}.
$$
Therefore, the birth rate at the nontrivial equilibrium (which necessarily exists if $R_0 > 1$ as discussed in Section \ref{section_EofSS}) is the unique positive solution $\bar{B}$ of the equation $\frac{\ln \left(1+B/(\mu z_0) \right)}{B/(\mu z_0)} = \frac{1}{R_0}.$ This allows an explicit expression for $\bar{B}$ in terms of the Lambert function $W_{-1}$ as
\begin{equation}\label{rationalbirthrate}
\bar{B} = \mu z_0 \left( -R_0 W_{-1}\bigg(-\frac{\exp(-1/R_0)}{R_0} \bigg) - 1 \right).
\end{equation}
Indeed, take  $z = - (1 + B /(\mu z_0))/R_0 < - 1/R_0$  in the preceding equation, which gives  $z e^{z} = - (1/R_0) e^{-1/R_0}$. Then the (only) solution to this equation is the Lambert function  $W_{-1}$ ( i.e., the inverse function of the (monotonously decreasing) function  $f(z) = z e^{z}$  restricted to the interval $(-\infty, -1)$)  evaluated at  $-(1/R_0) e^{-1/R_0}$.

More interestingly, an explicit expression can also be obtained for the density with respect to size in the steady state. Indeed, \eqref{size_equilibrium} gives in this case,
\begin{equation}\label{rational_size}
\bar{S} (a) = \int_0^a \frac{g_0}{1+\frac{\bar{B} e^{- \mu \tau}}{\mu z_0}} \mbox{d} \tau =   \frac{g_0}{\mu}  \ln \left(\frac{\mu z_0 e^{\mu a}  + \bar{B} }{\mu z_0 + \bar{B}} \right)   ,
\end{equation}
which leads to 
$$
\bar{S}^{-1} (x) = \frac{1}{\mu} \ln \left( \frac{(\mu z_0 + \bar{B}) e^{\frac{\mu}{g_0}
x} -\bar{B}}{\mu z_0} \right)
$$
and to
$$
\bar{S}'\big(\bar{S}^{-1}(x) \big) = g_0 \frac{(\mu z_0 + \bar{B}) \exp{\left(\frac{\mu}{g_0}x\right)} - \bar{B}}{(\mu z_0 + \bar{B}) \exp{\left(\frac{\mu}{g_0}x\right)}}.
$$
By \eqref{size_density} we finally obtain
\begin{equation}\label{rationaldensity}
\bar{u}(x) = \frac{\bar{B} \exp(-\mu \bar{S}^{-1} (x))}{\bar{S}'\big(\bar{S}^{-1}(x) \big)} =  \frac{\mu z_0 \bar{B} (\mu z_0 + \bar{B}) \exp\left(\frac{\mu}{g_0} x\right)}{g_0 \Big(\left(\mu z_0 + \bar{B}\right) \exp\left(\frac{\mu}{g_0} x\right) - \bar{B} \Big)^2}.
\end{equation}
Moreover, an easy integration gives the following expression for the population number above an individual of size $x$
\begin{equation}\label{popnumber}
\int_{x}^{\infty} \bar{u}(s) \mbox{d} s = \frac{\bar{B} z_0}{(\mu z_0+\bar{B})e^{\frac{\mu}{g_0}x}-\bar{B}}.
\end{equation}

\section{Concluding remarks}

The principle of linearised stability (PLS for short), widely used in the theory of ODEs, says that the stability of a stationary state is determined by the stability properties of the linearised system.  This principle has also been proved to hold in dynamical systems of infinite dimension with a ``semilinear'' structure (namely semilinear PDEs and DE, see \cite{henry1981,pazy1983,diekmann1995}) via the variation of constants formula. In this article we used the PLS to analyse rigorously the local stability of stationary birth rates of (\ref{scalar2}). As a consequence of such an analysis we found that for reasonable and rather general biological functional responses (see the hypotheses of Theorem \ref{estimate-theorem}), the non-trivial stationary birth rate of (\ref{scalar2}) is locally asymptotically stable.

The PLS,  as stated above,  cannot be applied to the PDE formulation presented in Appendix \ref{appendixPDE}. The reason is that, as explained in detail in \cite{BCDF2022}, the nonlinear semigroup associated to (\ref{pde}) is not differentiable, and hence it cannot be linearised. This does not mean, however, that, if the PDE system (\ref{pde}) is linearised ``formally'' around a stationary distribution $\bar{u}$, the stability of $\bar{u}$ cannot be determined from the stability of the linearised system. In fact we expect that such is possible, but a proof,  as far as we know,  is still missing.

As explained in \cite{BCDF2022}, a way to prove this result would be to establish an ``equivalence'' between orbits of the delay formulation (in the state space of weighted birth rate histories, i.e. $\mathcal{X}$) and orbits of the PDE formulation (in the state space of integrable functions of height, i.e. $L^1(x_m,\infty)$). By an ``equivalence'' we specifically mean to find a continuous function $\mathcal{L}_\text{DE}^\text{PDE}:\mathcal{X} \rightarrow L^1(x_m,\infty)$ mapping orbits in $\mathcal{X}$ to orbits in $L^1(x_m,\infty)$ and vice-versa (i.e. an analogous continuous function $\mathcal{L}_\text{PDE}^\text{DE}:L^1(x_m,\infty)\rightarrow \mathcal{X}$), so that stability results can be translated from one formulation to the other. In \cite{BCDF2022} we found that for these functions to exist, one needed to work in a (exponentially) weighted space of integrable functions of height, $L^1_w(x_m,\infty)$, where the proper value of $w$ depended on the weight $\rho$ chosen for $\mathcal{X}$ (working with the unweighted space $L^1(x_m,\infty)$ was possible if $\rho$ was chosen to be equal to the mortality rate $\mu$, since that implied $w=0$). In fact, in that paper the phase spaces for both the PDE and the DE included a component with information on the environmental condition. These additional components allowed to establish a surjective function (with the desired properties mentioned above) mapping states from the delay formulation to states of the PDE formulation (and vice-versa by taking a pseudoinverse of that function). As we are about to see, the analogous function associated to the (simpler) phase spaces used in this paper fails to be surjective (precluding any attempt of extending the results of \cite{BCDF2022} to the present work).

Natural candidates for $\mathcal{L}_\text{DE}^\text{PDE}$ and $\mathcal{L}_\text{PDE}^\text{DE}$ may be obtained using the biological interpretation of the functions involved in \eqref{scalar2} and \eqref{pde}. Indeed, take $\phi\in\mathcal{X}$ a birth rate history and $u_0\in L^1(x_m,\infty)$ a `corresponding' population height-distribution and define
\begin{equation}\label{defX}
\begin{aligned}
X(\tau;\phi):=S(-\tau,0;\phi) & =  x_m+\int_{0}^{-\tau} g\left(\int_\sigma^\infty \phi(\tau+\sigma-\alpha)e^{-\mu\alpha}\ud\alpha\right)\ud\sigma\\
& =  x_m+\int_{0}^{-\tau} g\left(\int_{-\infty}^\tau \phi(\theta)e^{-\mu(\tau+\sigma-\theta)}\ud\theta\right)\ud\sigma 
\end{aligned}
\end{equation}
for $\tau\in(-\infty,0]$ (i.e. the size at time 0 of an individual born at $\tau$ given the birth rate history $\phi$, see (\ref{heightFunction})) and $T(x;\phi)$, for $x\in[x_m,\infty)$, as the inverse of $X(\cdot;\phi)$ (which exists if $g$ is bounded and decreasing and gives the time at birth of an individual with size $x$ at time 0 given the birth rate history $\phi$). Then we have
\begin{equation}\label{conservationEquivalence}
\int_{x_m}^{x} u_0(x)\ud x = \int_{T(x;\phi)}^0 \phi(\theta)e^{\mu\theta}\ud\theta
\end{equation}
because being younger means being smaller, and hence the individuals smaller than $x$ must coincide with those born after $T(x;\phi)$ that have survived. Then, differentiation with respect to $x$ gives
\begin{equation}\label{L_DEtoPDE}
u_0(x) = -\phi(T(x;\phi))e^{\mu T(x;\phi)}T'(x;\phi), 
\end{equation}
which gives a natural candidate for $\mathcal{L}_\text{DE}^\text{PDE}$. Similarly, by rewriting (\ref{conservationEquivalence}) as
\[
\int_{x_m}^{X(\tau;\phi)} u_0(x)\ud x = \int_\tau^0 \phi(\theta)e^{\mu\theta}\ud\theta,
\]   
differentiation with respect to $\tau$ gives
\begin{equation}\label{L_PDE_to_DE}
u_0(X(\tau;\phi))X'(\tau;\phi) = -\phi(\tau)e^{\mu\tau}. 
\end{equation}
Unlike (\ref{L_DEtoPDE}), the above equation is problematic in that it does not give an explicit formula for $\phi$ in terms of $u_0$. It turns out that the above equation does not define implicitly $\phi\in\mathcal{X}$ for each $u_0\in L^1(x_m,\infty)$ (which is equivalent to say that $\mathcal{L}_\text{DE}^\text{PDE}$ defined through (\ref{L_DEtoPDE}) is not surjective). To see this choose, as a counterexample, $\mu=0$, $g(E)=1-E$ for $E<1/2$ (it doesn't matter what $g$ does for $E\geq 1/2$, besides being decreasing) and $u_0(x)=1$ for $x\in(x_m,x_m+1)$ and 0 otherwise. Then formula (\ref{defX}) simplifies to $X(\tau;\phi) = x_m - \tau g\left(\displaystyle\int_{-\infty}^{\tau}\phi(\theta)\ud\theta\right)$ and equation (\ref{L_PDE_to_DE}) implies
\[
\phi(\tau)=\frac{g\left(\displaystyle\int_{-\infty}^{\tau}\phi(\theta)\ud\theta\right)}{1-\tau g'\left(\displaystyle\int_{-\infty}^{\tau}\phi(\theta)\ud\theta\right)},
\]
if $X(\tau;\phi)<1$ and $\phi(\tau)=0$ otherwise. This forces the support of $\phi$ to be $(-1,0)$, so that $X(\tau;\phi) < x_m+1$ for $\tau\in(-1,0)$, and thus $\phi$ solves (\ref{L_PDE_to_DE}) only if it satisfies
\[
\phi(\tau)  = \frac{1-\displaystyle\int_{-1}^{\tau}\phi(\theta)\ud\theta}{1+\tau} 
\]
as long as $\displaystyle\int_{-1}^{\tau}\phi(\theta)\ud\theta <1/2$. Since the right hand side of this equation has a non-integrable singularity for $\tau\downarrow -1$, this relation contradicts that $\phi\in\mathcal{X}$. The fact that equation (\ref{L_PDE_to_DE}) fails to define a birth rate history in $\mathcal{X}$ as a function of $u_0$ means that there are reasonable population densities with respect to size (such as the indicator function used in the example) that cannot be obtained by prescribing an integrable birth rate history. 

As already mentioned, this situation deviates from what we had in \cite{BCDF2022}, where an explicit formula for $\mathcal{L}_\text{PDE}^\text{DE}$ was derived thanks to the additional environmental variable that was considered as part of the phase space (and somehow provided more room to play with). Since the scalar renewal equation presented in Section \ref{sectionDelayForm} was obtained precisely by expressing the environmental variable in terms of the birth rate history (and thus restricting the set of admissible environmental histories), a way to overcome this difficulty would be to work with an extended version of the delay formulation in which the environmental history is a proper element of the phase space (and thus there is also a delay equation for it). In addition, such an extended version would allow us to analyse more general environmental feedbacks. For instance environmental feedbacks of the form $E(x,t)=\int_x^\infty \alpha(y) u(y,t)dy$ (compare with (\ref{environmentalVariable})), where the impact of larger individuals depends on their size. Such situations cannot be formulated in terms of only a renewal equation for the birth rate. Indeed, since the environmental history is needed to give the size individuals will have in the future, the environmental condition felt by an individual is no longer determined only by the individuals born before him but it depends also on the environmental history itself. The drawback of an extended formulation is that then the environmental history $t \mapsto E(\cdot,t)$ takes values in an infinite dimensional space, which makes the analysis of the differentiability (analogue of Theorem \ref{differentiability}) much more involved (the theory to deal with these cases is developed in \cite{Diek2}). 

What could be the implications of such a non-equivalence between the two formulations, and specifically of the fact that there are population densities that cannot be obtained naturally from a birth rate history? It seems that the non-equivalence does not imply differences in the number of stationary states and attractors in general found in each formulation. In fact we expect a one-to-one correspondence between orbits in the $\omega$-limit sets of the two formulations (such a correspondence would be a consequence of the relation
between solutions of the RE and solutions of the PDE given in subsection \ref{secB1} of Appendix \ref{appendixPDE}). What might be affected by the non-equivalence is the stability behaviour of the corresponding $\omega$-limit sets. A priori (with what we have shown in this paper) we cannot rule out the possibility that a stationary population density of the PDE formulation is unstable, while the corresponding stationary birth rate of the delay formulation is stable. The reason is that there are states arbitrarily close to such a stationary population density that cannot be related to any birth history from a neighbourhood of the stationary birth rate. Further work is needed to rule out this kind of discrepancy between the two formulations (or, alternatively, to give a specific example where the discrepancy takes place, although we doubt that such an example exists).

\bmhead{Acknowledgments}

This work was partially supported by the research projects MT2017-84214C2-2-P and PID2021-123733NB-I00. We also thank the International Centre for Mathematical Sciences for financial support we received from the Research in Groups program during our stay at Edinburgh in July 2017. The first ideas for the present manuscript arose there and then.  We thank the reviewers for their careful reading of the manuscript and helpful comments.

\begin{appendices}

\section{Differentiability}\label{appendixB}
\begin{theorem}
\textbf{(Theorem \ref{differentiability})} Assume that $g$ and $\beta$ have a bounded and globally Lipschitzian first derivative (with common constant $2C$). Also assume that $g$ is bounded, positive and bounded away from $0$.
	Then the map $\mathcal{F}:\mathcal{X}\rightarrow\mathbb{R}$ defined in (\ref{scalar3}) is continuously differentiable with bounded derivative provided that the parameter $\rho$ in the definition of $\mathcal{X}$ satisfies $ \rho < \mu/5.$
\end{theorem}
	
	\begin{proof}
	First notice that the hypotheses imply the following estimate for any $z\geq 0$ and $h > -z$:
	
\begin{equation}\label{estimate}
\begin{array}{lr}
\vert g(z+h)-g(z)-g'(z)h\vert \\
 = \left\vert \int_z^{z+h} g'(s) ds - g'(z)h \right\vert
 = \left\vert \int_z^{z+h} \vert g'(s)-g'(z)\vert \, ds \right\vert \leq 2C \left\vert \int_z^{z+h} \vert z-s\vert ds \right\vert  =C h^2,
\end{array}
\end{equation}
and analogously for $\beta$.	
	
	The statement of the theorem amounts to showing that 
$$
\phi \rightarrow (\tilde{\mathcal{F}}(\phi)) (a) =   e^{-\mu a} \, \beta \bigg( x_m + \int_0^a g \big ( e^{- \mu(\tau -a)}
\int_{a}^{\infty} e^{-\mu s} \phi (-s) ds
\big) \, d\tau \bigg)
$$
is a continuously differentiable map from the positive cone of the Banach space
$$\mathcal{X}= \left\{ \phi \in L_{loc}^1(- \infty, 0):\left\vert \left\vert  \phi \right \vert  \right\vert_{\mathcal{X}} := \int_{-\infty}^0 e^{\rho s}\vert \phi(s)\vert ds < \infty \right\}$$
 to its dual identified with the Banach space 
$$\mathcal{X}' = \left\{f \in L_{loc}^{\infty}(0, \infty): \vert\vert f \vert\vert_{\mathcal{X}'} := \esssup_{a \in [0, \infty)} e^{\rho a} \vert f(a)\vert  < \infty  \right\}
$$ 
with the duality product  $\langle f, \phi \rangle = \int_0^{\infty} f(a) \phi(-a) \mbox{d} a .$

Indeed, we can write $\mathcal{F} (\phi) = \langle \tilde{\mathcal{F}} (\phi), \phi \rangle,$ and a rather general and straightforward argument gives, assuming differentiability of	 $\tilde{\mathcal{F}},$ 
\begin{equation} \label{differential}
D \mathcal{F} (\phi) \psi  =  \langle \tilde{\mathcal{F}}(\phi), \psi  \rangle + \langle \ D \tilde{\mathcal{F}}(\phi) \psi, \phi  \rangle.
\end{equation}
In particular, for $\phi=0,$ we have $ D \mathcal{F} (0) \psi  =  \langle \tilde{\mathcal{F}}(0), \psi  \rangle .$

Next we define three intermediate spaces of real valued continuous functions:
$$
Y = \left\{P \in C(T): \vert\vert P\vert\vert_{Y} := \sup_{(\tau,a) \in T} e^{-\rho a} \vert P(\tau,a)\vert < \infty  \right\}
$$
where $T=\{(\tau,a)\in \mathbb{R}^2: 0 \leq\tau\leq a < \infty\},$
$$
Z = \left\{v \in C(T): \vert\vert v\vert\vert _{Z} := \sup_{(\tau,a) \in T} e^{-\rho_1 a} \vert v(\tau,a)\vert  < \infty  \right\}
$$
with $\rho_1 >0$ to be chosen later,
$$
W = \left\{S \in C([0, \infty)): \vert\vert S\vert\vert_{W} := \sup_{a \in [0, \infty)} e^{-\rho_2 a}  \vert S(a)\vert < \infty  \right\}
$$
with $\rho_2>0$ to be chosen later;
and four maps:
$$
\mathcal{L}_1 : \mathcal{X} \rightarrow Y \,\, \text{defined by} \,\, (\mathcal{L}_1 \phi) (\tau,a) =  e^{- \mu(\tau -a)}
\int_{a}^{\infty} e^{-\mu s} \phi (-s) \mbox{d}s,
$$
$$
\mathcal{G}: Y \rightarrow Z \,\, \text{defined by} \,\, \mathcal{G}(P) = g \circ P,
$$	
$$
\mathcal{L}_2 : Z \rightarrow W \,\, \text{defined by} \,\, (\mathcal{L}_2 v) (a) = 
x_m + \int_{0}^{a} v(\tau,a) \mbox{d}\textcolor{blue}{\tau}
$$
and
$$
\mathcal{B}: W_{+} \rightarrow \mathcal{X}' \,\, \text{defined by} \,\, \mathcal{B}(S)(a) = e^{- \mu a} \, (\beta \circ S) (a),
$$
($W_{+}$ meaning the positive cone of $W$) in such a way that (at least formally) $\tilde{\mathcal{F}} = \mathcal{B} \circ \mathcal{L}_2 \circ \mathcal{G} \circ \mathcal{L}_1.$ Then the claim will follow from the chain rule provided we prove that the four maps are well defined and continuously differentiable with bounded derivative.

\textbf{Step 1. $\mathcal{L}_1$ is bounded linear provided that $\rho \leq \mu$.}

We have
$$
\sup_{(\tau,a) \in T} e^{-\rho a} \left\vert e^{- \mu(\tau -a)}\int_{a}^{\infty} e^{-\mu s} \phi (-s) ds \right\vert  =  \sup_{(\tau,a) \in T} e^{-\rho a} \left\vert e^{- \mu(\tau -a)}\int_{-\infty}^{-a} e^{\mu s} \phi (s) ds \right\vert 
$$
$$
\leq \sup_{a \geq 0} \int_{-\infty}^{-a} e^{(\mu-\rho)(a+s) } e^{\rho s} \vert\phi(s)\vert ds \leq \int_{- \infty}^0 e^{\rho s} \vert\phi(s)\vert ds,
$$
since $(\mu-\rho)(a+s) \leq 0$ in the last but one integral. Thus,
$
||\mathcal{L}_1 \phi||_{Y} \leq|| \phi||_{\mathcal{X}} 
.$

\textbf{Step 2. $\mathcal{G}$ is continuously differentiable with bounded derivative provided that $2 \rho \leq \rho_1$.} \\
 $\mathcal{G}$ is well defined because $g$ is bounded and continuous.\\
 Let $P \in Y$ and $Q \in Y$ such that  $\vert\vert Q\vert\vert_{Y} = 1$, which implies $\vert Q(\tau,a)\vert  \leq e^{\rho a}.$ \\
 We start by proving that $Q \rightarrow g'(P(\cdot)) Q(\cdot)$ defines a bounded linear map $\mathcal{Y} \rightarrow \mathcal{Z}$ with norm bounded independently of $P$:
 \begin{equation*}
 	\begin{aligned}
 \sup_{\left\vert \left\vert  Q  \right\vert \right\vert_{Y} = 1} \left\vert \left\vert  g'(P(\cdot)) Q (\cdot) \right\vert \right\vert_{Z} & = \sup_{\left\vert \left\vert  Q  \right\vert  \right\vert_{Y} = 1} \sup_{z \in T} e^{-\rho_1 a} \left\vert g'(P(z)) Q(z)  \right\vert \\
&  \leq  \left\vert \left\vert  g' \right\vert \right\vert_{\infty}  \sup_{z \in T} e^{- \rho a} \left\vert Q(z) \right\vert =  \left\vert \left\vert  g' \right\vert \right\vert_{\infty}.
 \end{aligned}
 \end{equation*}

  Moreover, we can write, setting $z =(\tau,a),$ and using \eqref{estimate},
 \begin{equation*}
 	\begin{aligned}
 e^{- \rho_1 a} \vert g(P(z)+ \varepsilon Q(z))-g(P(z))-g'(P(z)) \varepsilon Q(z)\vert  & \leq C \varepsilon^2  e^{- \rho_1 a} \vert Q(z) \vert^2 \\
 & \leq C \varepsilon^2 e^{(2 \rho-\rho_1)a} \leq C \varepsilon^2,
 \end{aligned}
\end{equation*}
i.e., 
$$\left\vert\left\vert \mathcal{G}(P + \varepsilon Q) - \mathcal{G}(P) - g'(P(\cdot)) \varepsilon Q(\cdot)  \right\vert \right\vert_{Z}\leq  C \varepsilon^2 .$$ 
Therefore, $(D\mathcal{G}(P) Q)(z) := g'(P(z))Q(z)$ is the Fréchet derivative of $\mathcal{G}$ at the point $P$, its norm is uniformly bounded by $\left\vert \left\vert g' \right\vert \right\vert_{\infty}$; and it is (uniformly) continuous: for $Q \in Y$ with norm $1$ we have
\begin{equation*}
	\begin{aligned}
& \left\vert \left\vert D\mathcal{G}(P_1) Q - D\mathcal{G}(P_2) Q \right\vert\right\vert_{Z} = \sup_{z \in T} e^{- \rho_1 a} \left\vert\big(g'(P_1(z))-g'(P_2(z))\big) Q(z) \right\vert \\
& \leq 2C \sup_{z \in T} e^{- \rho_1 a} \vert P_1(z)-P_2(z)\vert \,\vert Q(z)\vert \leq 2C \sup_{a \geq 0} e^{(-\rho_1 + 2 \rho) a} \vert\vert P_1-P_2\vert\vert_Y \leq 2C \vert\vert P_1-P_2\vert\vert_Y.
\end{aligned}
\end{equation*}

\textbf{Step 3. $\mathcal{L}_2$ is a positive continuous affine map provided that $0<\rho_1 < \rho_2$.}\\
It suffices to see,

\begin{equation*}
\begin{array}{ll}
\left\vert\left\vert  \mathcal{L}_2 v  - x_m   \right\vert \right\vert_W =
\displaystyle\sup_{a \in [0, \infty)} e^{- \rho_2 a} \left\vert \int_0^a v(\tau,a) d \tau \right\vert \\
\leq \displaystyle\sup_{a \in [0, \infty)}  e^{- \rho_2 a} \int_0^a e^{\rho_1 a} \vert\vert v \vert\vert_Z \, \mbox{d} \tau = \displaystyle\sup_{a \in [0, \infty)} a e^{(-\rho_2 + \rho_1) a} \vert\vert v \vert\vert_Z \leq \frac{1}{e (\rho_2-\rho_1)} \vert\vert v \vert\vert_Z.
\end{array}
\end{equation*}

\textbf{Step 4.\\ $\mathcal{B}$ is continuously differentiable with bounded derivative provided that $\rho + 2 \rho_2 \leq \mu$.}\\
First notice that the assumptions on $\beta$ imply that there exist positive constants $C_1$ and $C_2$ such that $\beta(s) \leq C_1 + C_2 s.$ Thus $\mathcal{B}$ is well defined: for $S \in W_{+ },$ since $\vert S(a)\vert  \leq e^{\rho_2 a} \vert\vert S \vert\vert_{W},$
\begin{equation*}
	\begin{aligned}
e^{\rho a} \vert e^{-\mu a} \beta(S(a)) \vert &  \leq e^{(\rho-\mu)a} (C_1+C_2 \vert S(a)\vert ) \\ 
& \leq e^{(\rho-\mu)a} (C_1 + C_2 e^{2 \rho_2 a} \vert\vert S\vert\vert_W) \leq C_1 + C_2 \vert\vert S\vert\vert_W.
\end{aligned}
\end{equation*}
As in Step 2, let us prove that $R \rightarrow e^{- \mu \cdot} \beta'(S(\cdot)) R(\cdot)$ defines a bounded linear map $W \rightarrow \mathcal{X'}$ with norm bounded independently of $S$:

\begin{equation*}
	\begin{aligned}
 \sup_{\left\vert \left\vert R \right\vert \right\vert_{W} = 1} \left\vert \left\vert  e^{- \mu \cdot} \beta'(S(\cdot)) R(\cdot) \right\vert \right \vert_{\cal{X}'} = &  \sup_{\left\vert\left\vert R \right\vert \right\vert_{W} = 1} \sup_{a \geq 0} e^{\rho a} \left\vert e^{- \mu a} \beta'(S(a)) R(a)  \right\vert \\
  \leq & \sup_{a \geq 0} e^{(\rho + \rho_2 - \mu) a} 
\left\vert \left\vert\beta' \right\vert \right\vert_{\infty} \leq \left\vert \left\vert\beta' \right\vert \right\vert_{\infty}.
\end{aligned}
\end{equation*}

Let us now proceed to show that $\mathcal{B}$ is  differentiable: Let $S \in W_{+}$ and $R \in W$ with norm equal to 1, which implies $\vert R(a)\vert < e^{\rho_2 a}$. Then, for $\varepsilon$ small enough, $S + \varepsilon R \in
W_{+}$.
Then we can write, using the $\beta$-variant of \eqref{estimate},

\begin{equation*}
\begin{array}{ll}
 e^{\rho a} \left\vert e^{-\mu a}\beta(S(a)+ \varepsilon R(a))-e^{-\mu a}\beta(S(a))-e^{-\mu a}\beta'(S(a)) \varepsilon R(a) \right\vert \\
 \leq C \varepsilon^2 e^{(\rho-\mu + 2\rho_2)a} \leq C \varepsilon^2,
\end{array}
\end{equation*}
proving that $\big(D\mathcal{B} (S) R \big) (a) := e^{- \mu a} \beta'(S(a)) R(a)$ is the Fr\'{e}chet derivative of  $\mathcal{B}$ at the point $S,$ with norm uniformly bounded by $\left\vert \left\vert\beta' \right\vert \right\vert_{\infty} $\\ 
We also show that the derivative is continuous as in Step 2. Let $R \in W$ with norm equal to 1. We have
\begin{equation*}
	\begin{aligned}
\left\vert \left\vert D\mathcal{B}(S_1) R - D\mathcal{B}(S_2) R \right\vert \right\vert_{\mathcal{X}'} = & \sup_{a \in [0, \infty)} e^{(\rho - \mu) a}  \left\vert \left( \beta'(S_1(a))-\beta'(S_2(a)) \right) R(a) \right\vert \\
 \leq & 2C \sup_{a \in [0, \infty)} e^{(\rho - \mu) a} \vert S_1(a)-S_2(a)\vert \,\vert R(a)\vert  \\ 
 \leq & 2C \sup_{a \geq 0} e^{(\rho - \mu + 2 \rho_2) a} \vert\vert S_1-S_2\vert\vert_W \leq 2C \vert\vert S_1-S_2\vert\vert_W .
 \end{aligned}
\end{equation*}
Finally, given any $\rho \in \left(0,\frac{\mu}{5}\right)$ we can take $\rho_1 = 2 \rho$ (fulfilling the assumption of Step 2) and $\rho_2 = \frac{\mu - \rho}{2} > 2 \rho$ (fulfilling the assumption of Step 3 and that of Step 4 since then $2 \rho_2 + \rho = \mu)$  to conclude the proof.
\end{proof}

	As a consequence, the chain rule gives, taking into account that  $\mathcal{L}_1$ is linear and $\mathcal{L}_2$ is affine, 
$$ D \tilde{\mathcal{F}} (\phi) \psi = D(\mathcal{B} \circ \mathcal{L}_2 \circ \mathcal{G} \circ \mathcal{L}_1)(\phi) \psi = D\mathcal{B}(\mathcal{L}_2 \mathcal{G}(\mathcal{L}_1 \phi)) \, (\mathcal{L}_2 - x_m) \, D \mathcal{G}(\mathcal{L}_1 \phi) \, \mathcal{L}_1 \psi.
$$
Since we are interested in linearisation around steady states, we can restrict to evaluation of the differential on constant functions $\bar{B}.$ So, we compute, sequentially:
\begin{equation*}
	\begin{aligned}
\mathcal{L}_1 \psi \, (\tau,a) = & e^{-\mu(\tau-a)} \int_a^{\infty}e^{-\mu s} \psi(-s) \mbox{d}s, \\
\mathcal{L}_1 \bar{B}\, (\tau,a) = & \bar{B} e^{-\mu \tau}/\mu, \\
D \mathcal{G} (\mathcal{L}_1 \bar{B}) \mathcal{L}_1 \psi \, (\tau, a) = &  g'(\bar{B}e^{-\mu \tau}/\mu) e^{-\mu(\tau-a)} \int_a^{\infty}e^{-\mu s} \psi(-s) \mbox{d} s, \\
(\mathcal{L}_2 - x_m) \, D \mathcal{G} (\mathcal{L}_1 \bar{B}) \mathcal{L}_1 \psi \, (\tau, a) = & \int_0^a  g'(\bar{B}e^{-\mu \tau}/\mu) e^{-\mu(\tau-a)} \int_a^{\infty}e^{-\mu s} \psi(-s) \mbox{d} s \, \mbox{d} \tau \\ 
= &: h(a),
\end{aligned}
\end{equation*}
and, also,
$$
\mathcal{L}_2\, \mathcal{G} (\mathcal{L}_1 \bar{B}) = x_m + \int_0^a g(\bar{B} e^{-\mu \tau}/\mu) \mbox{d} \tau (= \bar{S} (a)),
$$
where, in the last equality we assumed, furthermore, that $\bar{B}$ is not only a constant function, but a steady state (see \eqref{size_equilibrium}). Therefore,
\begin{equation*}
	\begin{aligned}
 D \tilde{\mathcal{F}} (\bar{B}) \psi = & D\mathcal{B}(\mathcal{L}_2 \mathcal{G}(\mathcal{L}_1 \bar{B})) \, (\mathcal{L}_2 - x_m) \, D \mathcal{G}(\mathcal{L}_1 \bar{B}) \, \mathcal{L}_1 \, \psi (a) \\
= & D\mathcal{B} ( \mathcal{L}_2 \, \mathcal{G} (\mathcal{L}_1 \bar{B})) h(a) = 
 e^{-\mu a} \beta'(\bar{S}(a)) h(a)  \\ 
 = & e^{-\mu a} \beta'(\bar{S}(a)) \, \int_0^a  g'(\bar{B}e^{-\mu \tau}/\mu) e^{-\mu(\tau-a)} \int_a^{\infty}e^{-\mu s} \psi(-s) \mbox{d} s \, \mbox{d} \tau.
\end{aligned}
\end{equation*}
Finally, we will have,
$$
\langle \ D \tilde{\mathcal{F}}(\bar{B}) \psi, \bar{B}  \rangle = \int_0 ^{\infty} e^{-\mu a} \beta'(\bar{S}(a)) \, \int_0^a  g'(\bar{B}e^{-\mu \tau}/\mu) e^{-\mu(\tau-a)} \int_a^{\infty}e^{-\mu s} \psi(-s) \mbox{d} s \, \mbox{d} \tau \bar{B} \mbox{d} a,
$$
which, together with \eqref{differential}, gives \eqref{nontrlinear}.

\section{The PDE formulation}\label{appendixPDE}

In this appendix we include a series of results showing that the PDE formulation is tightly related to the delay formulation (as it should be since both models are built from a description of the same biological processes). In subsection \ref{secB1} we show that one can solve the PDE problem by solving a scalar RE (with integration from 0 to $t$) for the population birth rate $B$ and that the large time limiting form of this equation is exactly (\ref{scalar}). In subsection \ref{secB2} we show that the condition characterising the existence of non-trivial steady states of (\ref{pde}) coincides with (\ref{def_R}) (in addition a formula for the non-trivial stationary population size-density is given). Finally in subsection \ref{secB3} we show that the formal linearisation of system (\ref{pde}) leads to the characteristic equation (\ref{DE-char-new_1}).

\subsection{Solution of the PDE in terms of a renewal equation}\label{secB1}

The solution of (\ref{pde}) can be written as the sum of two terms: the first considers the individuals born between $0$ and $t$ and the second considers the individuals that already exist at time 0, i.e. those reflected in the initial population density $u_0(x)$.

First notice that at time 0,
\[
\bar{E}(\xi)=\int_\xi^\infty u_0(\eta)d\eta
\]
gives the number of individuals with size larger than $\xi$, while at time $\tau$
\[
\tilde{E}(\tau)=\left(\int_0^\tau B(\sigma)e^{\mu\sigma}d\sigma + \int_0^\infty u_0(\eta)d\eta\right)e^{-\mu\tau}
\]
gives the number of individuals with size larger than $x_m$.   Since mortality is constant the number of individuals in this cohort decreases exponentially with rate $\mu$ as time increases.  As a consequence, the size at time $t$ of an individual with size $\xi$ at time 0 is
\[
X(t,0,\xi)=\xi+\int_0^t g(\bar{E}(\xi)e^{-\mu\sigma})d\sigma
\]
and the size at time $t$ of an individual born at time $\tau>0$ with $0<\tau<t$ is
\[
X(t,\tau,x_m)=x_m+\int_0^{t-\tau} g\left(\tilde{E}(\tau)e^{-\mu\sigma}\right)d\sigma.
\]
So the birth rate has to satisfy the renewal equation
\[
B(t)=B_{\text{dsc}}(t)+B_{\text{fnd}}(t)
\]
where
\[
B_{\text{dsc}}(t) = \int_0^t \beta(X(t,\tau,x_m))B(\tau)e^{-\mu(t-\tau)}d\tau
\]
is the birth rate associated to the descendants of the founder population and
\[
B_{\text{fnd}}(t) = \int_0^t \beta(X(t,0,\xi))u_0(\xi)d\xi\,e^{-\mu t}
\]
is the known birth rate associated to the founder population. Once we solve the renewal equation constructively, we can obtain an explicit expression for the (weak) solution of the PDE by integrating along characteristics.

Note that $B_{\text{fnd}}(t)$ tends to 0 exponentially as $t\rightarrow \infty$. By changing $\tau$ to $a$ with $t-\tau=a$ we can rewrite
\[
B_{\text{dsc}}(t) = \int_0^t \beta(X(t,t-a,x_m))B(t-a)e^{-\mu a}da.
\]
Now note that
\[
X(t,t-a,x_m)=x_m+\int_0^a g(\tilde{E}(t-a)e^{-\mu\tau})d\tau
\]
and
\[
\tilde{E}(t-a)=\int_0^{t-a} B(\eta)e^{\mu\eta}d\eta \,e^{-\mu(t-a)} + \int_0^\infty u_0(\eta)d\eta \, e^{-\mu(t-a)}
\]
where the second summand at the right hand side tends exponentially to 0 as $t\rightarrow \infty$. Since this term represents the founder population that remains at time $t$, let us refer to it as $P_\text{fnd}(t)$. Next, by using the transformation $\eta=t-s$ we have
\[
\tilde{E}(t-a)=\int_a^t B(t-s)e^{-\mu(s-a)}ds + P_\text{fnd}(t) = e^{\mu a}\int_a^t B(t-s)e^{-\mu s} ds + P_\text{fnd}(t).
\]
Now note that, by ignoring $B_{\text{fnd}}(t)$ and $P_{\text{fnd}}(t)$ and by replacing the upper integration boundary $t$ in the last integral by $\infty$, we obtain (\ref{scalar2}).

\subsection{Existence and characterization of non-trivial steady states}\label{secB2}

To establish criteria for the existence of non-trivial steady states $\bar{u}$ in the PDE formulation is apparently more complex than what we had to do for the delay formulation in Section \ref{section_EofSS}.

Let us first concentrate on the ordinary differential equation which arises from the first and the third equations in \eqref{pde} when one assumes that $\bar{u}$ only depends on $x$. This leads to the following second order ordinary differential equation for $E(x):=\int_{x}^{\infty}\bar{u}(s)ds$,
$$
\frac{d}{dx} \left(g(E(x)) E'(x) \right) + \mu E'(x) = 0,
$$
or, equivalently, to 
$$g(E(x)) E'(x)  + \mu E(x) = C,$$
 for some constant $C.$ Since $E(x)$ tends to $0$ when $x$ tends to $\infty,$ $C$ has to coincide with (minus) the flux of individuals leaving the system at infinity: $C =\displaystyle \lim_{x \rightarrow \infty} g(E(x))E'(x) = -g(0)\lim_{x \rightarrow \infty} \bar{u}(x)$ and so it has to be $0$ (since otherwise $\displaystyle\lim_{x \rightarrow \infty} \bar{u}(x) = - C/g(0)\neq 0$ and $\bar{u}$ would not be integrable).
 Therefore we look for solutions of the differential equation
 $$
 \frac{dE}{dx}(x) = -\mu \frac{E(x)}{g(E(x))}
 $$
 with initial condition $E(x_m)=N$ (the total population size) and such that $\displaystyle\lim_{x \rightarrow \infty} E(x)=0.$ Equivalently,
 $$
 \int_{E(x)}^N \frac{g(z)}{z} \ud z = \mu (x-x_m).
 $$
 If $G$ is a primitive of $g(z)/z,$ the previous equation reads
 $$
 G(N)-G(E(x)) = \mu (x-x_m),
 $$
 which, can be rewritten as
 $$
 E(x) = G^{-1} \left( G(N) - \mu (x-x_m)\right).
 $$
It follows that
\begin{equation}\label{new_steady_state}
\begin{aligned}
 \bar{u}(x) = & - \frac{d}{dx}\left( G^{-1} \left( G(N) - \mu(x-x_m) \right) \right) \\
  = & \frac{\mu}{G'\left( G^{-1} \left( G(N) - \mu (x-x_m) \right) \right)} = \mu \frac{G^{-1}\left(G(N) - \mu (x-x_m)\right)}{g\left(G^{-1}\left(G(N) - \mu (x-x_m)\right) \right)}. 
 \end{aligned}
\end{equation}
Since $\bar{u}(x_m)= \frac{\mu N}{g(N)}$ we have
 $$
 g(E(x_m)) \bar{u}(x_m) = g(N) \bar{u}(x_m) = \mu N.
 $$
 Therefore, using the boundary condition, a non-trivial steady state (given by \eqref{new_steady_state}) does exist if and only if a positive number $N$ exists such that
 \begin{equation}\label{compare_R0}
  N = \int_{x_m}^{\infty} \beta(x) \frac{G^{-1}\left(G(N) - \mu (x-x_m)\right)}{g\left(G^{-1}\left(G(N) - \mu (x-x_m)\right) \right)} \ud{x}.	
 \end{equation}
 This turns out to be equivalent to \eqref{def_R} with $N = B/\mu $. Indeed, we can write
 \begin{equation*}
 	\begin{aligned}
R(B)= & \int_0^{\infty} \beta \left(x_m + \int_0^a g \left( B
\frac{e^{-\mu \tau}}{\mu} \right) \, \mbox{d}\tau \right)\,\, e^{-\mu a}
\, \mbox{d}a \\
=  &\int_0^{\infty} \beta \left(x_m + \int_{N e^{- \mu a}}^N \frac{ g \left(z\right)}{\mu z} \, \mbox{d}z \right)\,\, e^{-\mu a}
\, \mbox{d}a  \\
 = & \int_0^{\infty} \beta \left(x_m  + \frac{G(N)-G(N e^{-\mu a})}{\mu} \right) e^{-\mu a}
 \, \mbox{d}a  \\ 
 = & \frac{1}{N}\int_{x_m}^{\infty} \beta(x) \frac{G^{-1}(G(N) - \mu (x-x_m))}{g\left(G^{-1}(G(N) - \mu (x-x_m) \right)} \mbox{d}x,
\end{aligned}
\end{equation*}
 where in the second equality we performed the change of variables $z = B
 \frac{e^{-\mu \tau}}{\mu}$, and in the fourth one, the change of variables $x = x_m + \frac{G(N)-G(N e^{-\mu a})}{\mu}.$
See Section \ref{part-cases} where a particular case is developed and where an explicit expression for a primitive $G$ is available.

\subsection{Linearisation in the PDE formulation}\label{secB3}

The (formal)  linearisation of the PDE \eqref{pde} around the steady state $u_*$ is very economical as it simply reads (note that $g'$ below stands for the derivative of $g$ with respect to its argument $E$)
\begin{equation}\label{lineq1}
\begin{aligned}
v_t(x,t) +\left(g(E_*(x))v(x,t)+g'(E_*(x))u_*(x)\int_x^\infty v(y,t)\,\ud y\right)_x= &-\mu v(x,t), \\
g(E_*(x_m))v(x_m,t)+g'(E_*(x_m))u_*(x_m)\int_{x_m}^\infty v(x,t)\,\ud x= &\int_{x_m}^\infty \beta(x)v(x,t)\,\ud x.
\end{aligned}
\end{equation}
Substituting $v(x,t)=e^{\lambda t}V(x)$ into \eqref{lineq1} we have
\begin{equation}\label{lineq2}
\begin{aligned}
\left(g(E_*(x))V(x)+g'(E_*(x))u_*(x)\int_x^\infty V(y)\,\ud y\right)_x= &-(\lambda+\mu) V(x), \\
g(E_*(x_m))V(x_m)+g'(E_*(x_m))u_*(x_m)\int_{x_m}^\infty V(x)\,\ud x= &\int_{x_m}^\infty \beta(x)V(x)\,\ud x.
\end{aligned}
\end{equation}
Therefore, $\lambda\in\mathbb{C}$ is an eigenvalue, if and only \eqref{lineq2} admits a solution $V\not\equiv 0$. 
We also note that although the size domain is unbounded, it can be shown, using certain properties of the governing linear semigroup,  that the part of the spectrum of the semigroup generator in the half plane $\left\{z\in\mathbb{C}\,  \vert \, \text{Re}(z)>-\mu\right\}$ contains only eigenvalues , see e.g. \cite[Sect.4.]{FH3} for more details, and therefore (linear) stability can indeed be characterized by the leading eigenvalue of the semigroup generator.

For the trivial steady state $u_*\equiv 0$ the left hand side of (\ref{lineq2}) has only local terms and therefore easily leads to the characteristic equation
\[
g(0)=\int_{x_m}^\infty \beta(x) e^{-\frac{\lambda+\mu}{g(0)}(x-x_m)}dx,
\]
which is exactly what one gets by inserting $y(t)=e^{\lambda t}$ into (\ref{linSolAt0}) and making the change of variables $x=x_m+g(0)a$. Therefore, the stability of $u_*\equiv 0$ is characterized by the net reproduction number ($R$ evaluated at the zero steady state, or the virgin environment as we previously referred to), as expected. That is,
if
\begin{equation*}
R(0)=(R_0=)\displaystyle\int_{x_m}^\infty\frac{\beta(x)}{g(0)}e^{-\frac{\mu}{g(0)}(x-x_m)}\,\ud x>1,
\end{equation*}
then $u_*\equiv 0$ is unstable; while $R(0)<1$ implies that the trivial steady state is asymptotically stable.

To deduce the characteristic equation we integrate the first equation of \eqref{lineq2} from $x$ to $\infty$, to obtain
\begin{equation}\label{lineq3}
g(E_*(x))V(x)+g'(E_*(x))u_*(x)\int_x^\infty V(y)\,\ud y=(\lambda+\mu)\int_x^\infty V(y)\,\ud y.
\end{equation}
Substituting $x=x_m$ into \eqref{lineq3} and combining it with the second equation in \eqref{lineq2} yields
\begin{equation}\label{lineq3-2}
(\lambda+\mu)\int_{x_m}^\infty V(x)\,\ud x=\int_{x_m}^\infty \beta(x)V(x)\,\ud x.
\end{equation}
Note that $\lambda\in\mathbb{C}$ is an eigenvalue if and only if \eqref{lineq3}-\eqref{lineq3-2} admits a solution $V\not\equiv 0$. To see for which $\lambda$ this is possible let us introduce
\begin{equation}\label{lineq-mod1}
H(x):=\int_x^\infty V(y)\,\ud y
\end{equation}
as unknown so that (\ref{lineq3}) boils down to the differential equation 
\[
-g(E_*(x))H'(x)+g'(E_*(x))u_*(x)H(x)=(\lambda+\mu)H(x),
\]
whose solution is
\begin{equation}\label{lineq-mod2}
H(x)=H(x_m)\exp\left(\int_{x_m}^x\frac{g'(E_*(r))u_*(r)-(\lambda+\mu)}{g(E_*(r))}\,\ud r\right)=H(x_m)\pi(x,\lambda),
\end{equation}
where we defined
\begin{equation}
\pi(x,\lambda):=\exp\left(\int_{x_m}^x\frac{g'(E_*(r))u_*(r)-(\lambda+\mu)}{g(E_*(r))}\,\ud r\right),\quad \lambda\in\mathbb{C},\,\,x\in[x_m,\infty).
\end{equation}
Then, substitution of (\ref{lineq-mod2}) into (\ref{lineq3-2}) via (\ref{lineq-mod1}) yields the characteristic equation 
\begin{equation}\label{lineq-mod3}
(\lambda+\mu) = - \int_{x_m}^\infty \beta(x)\frac{\partial}{\partial\, x}\pi(x,\lambda)\,\ud x.
\end{equation}
This equation can be rewritten (assuming that $\beta$ is differentiable) as
\begin{equation}\label{lineq6}
(\lambda+\mu)=-\beta(\infty)\pi(\infty,\lambda)+\beta(x_m)+\int_{x_m}^\infty \beta'(x)\pi(x,\lambda)\,\ud x.
\end{equation}

Next note that since $E_*'(x)=-u_*(x)$ we have 
\begin{align*}
\pi(x,\lambda)= & \exp\left(\int_{x_m}^x\frac{g'(E_*(r))u_*(r)-(\lambda+\mu)}{g(E_*(r))}\,\ud r\right) \\
= & \exp\left(-\int_{x_m}^x\frac{\lambda+\mu}{g(E_*(r)}\,\ud r\right)\exp\left(-\int_{x_m}^x\frac{\frac{\ud}{\ud\, r}(g(E_*(r))}{g(E_*(r))}\,\ud r\right)\\
= & \exp\left(-\int_{x_m}^x\frac{\lambda+\mu}{g(E_*(r)}\,\ud r\right)\frac{g(E_*(x_m))}{g(E_*(x))}.
\end{align*}
Then, if $\mu>\displaystyle\sup_{x\geq x_m}\left\{g'(E_*(x))u_*(x)\right\}$ (for example if $g'\le 0$) one has $\pi(\infty,\lambda)=0$ for every $\lambda\in\mathbb{C}$, and the characteristic equation reduces to 
\begin{equation}\label{lineq7}
\lambda+\mu=\beta(x_m)+\int_{x_m}^\infty \beta'(x)\exp\left(-\int_{x_m}^x\frac{\lambda+\mu}{g(E_*(r)}\,\ud r\right)\frac{g(E_*(x_m))}{g(E_*(x))}\,\ud x.
\end{equation}

Now let us rewrite equation \eqref{lineq7} such that the integration variable is age $a$, which will show that it is the characteristic equation \eqref{DE-char-new_1} in disguise.  Using that we have 
\begin{equation*}
\begin{aligned}
\frac{\ud a}{\ud x}(x)= & \frac{1}{g(E_*(x))},\quad \bar{S}(a)=\int_0^a g\left(\bar{B}\frac{e^{-\mu\tau}}{\mu}\right)\,\ud\tau=\Gamma^{-1}(a)=x, \\ 
\Gamma(x):= & \int_0^x\frac{1}{g(E_*(r))}\,\ud r, \quad E_*(x)=\bar{B} \frac{e^{-\mu\Gamma(x)}}{\mu},
\end{aligned}
\end{equation*}
equation \eqref{lineq7} can be rewritten as
$$
\lambda+\mu=\beta(x_m)+g\left(\frac{\bar{B}}{\mu}\right)\int_0^\infty \beta'(\bar{S}(a))e^{-(\lambda+\mu)a}\,\ud a,
$$
which is identical to the characteristic equation \eqref{DE-char-new_1} that was deduced from the delay formulation.




\end{appendices}


\bibliography{sn-bibliography}


\end{document}